\documentclass{article}
\usepackage[utf8]{inputenc}
\usepackage{fullpage}
\usepackage{xcolor}
\usepackage{graphicx}
\usepackage{amsmath}
\usepackage{amsthm}
\usepackage{amsfonts, amssymb}
\usepackage{cite}

\usepackage{amsmath}
\usepackage{amsfonts}
\usepackage{amssymb}\newcommand{\Co}{{\mathcal{C}}}
\newtheorem{theorem}{Theorem} 
\newtheorem{lemma}{Lemma}
\newtheorem{remark}{Remark}
\newtheorem{property}{Property}
\newtheorem{proposition}{Proposition}
\newtheorem{corollary}{Corollary}
\newtheorem{claim}{Claim}
\newtheorem{definition}{Definition}

\newenvironment{lemma-repeat}[1]{\begin{trivlist}
		\item[\hspace{\labelsep}{\bf\noindent Lemma \ref{#1} }]\em }%
	{\end{trivlist}}
\newenvironment{theorem-repeat}[1]{\begin{trivlist}
		\item[\hspace{\labelsep}{\bf\noindent Theorem \ref{#1} }]\em }%
	{\end{trivlist}}

\newcommand{\remove}[1]{}

\DeclareMathOperator{\false}{{\scriptstyle{FALSE}}}
\DeclareMathOperator{\true}{{\scriptstyle{TRUE}}}

\newcommand{\size}[1]{\ensuremath{\left|#1\right|}}
\newcommand{\set}[1]{\left\{ #1 \right\}}

\newcommand{\C}{{\mathcal{C}}}
\begin{document}
	\date{}
	\begin{titlepage}
		\title{Beyond Alice and Bob: Improved Inapproximability for Maximum Independent Set in CONGEST}
		\author{Yuval Efron\footnote{efronyuv@gmail.com}\\Technion\and Ofer Grossman\footnote{ofer.grossman@gmail.com
		}\\MIT\and Seri Khoury\footnote{seri\textunderscore khoury@berkeley.edu}\\UC Berkeley}	
		\maketitle
		
		\begin{abstract}
			By far the most fruitful technique for showing lower bounds for the CONGEST model is reductions to two-party communication complexity. This technique has yielded nearly tight results for various fundamental problems such as distance computations, minimum spanning tree, minimum vertex cover, and more. 
			
			In this work, we take this technique a step further, and we introduce a framework of reductions to $t$-party communication complexity, for every $t\geq 2$. Our framework enables us to show improved hardness results for maximum independent set. Recently, Bachrach et al.[PODC 2019] used the two-party framework to show hardness of approximation for maximum independent set. They show that finding a $(5/6+\epsilon)$-approximation requires $\Omega(n/\log^6 n)$ rounds, and finding a $(7/8+\epsilon)$-approximation requires $\Omega(n^2/\log^7 n)$ rounds, in the CONGEST model where $n$ in the number of nodes in the network. 
			
			We improve the results of Bachrach et al. by using reductions to multi-party communication complexity. Our results:
			\begin{enumerate}
				\item Any algorithm that finds a $(1/2+\epsilon)$-approximation for maximum independent set in the CONGEST model requires $\Omega(n/\log^3 n)$ rounds.
				\item Any algorithm that finds a $(3/4+\epsilon)$-approximation for maximum independent set in the CONGEST model requires $\Omega(n^2/\log^3 n)$ rounds.	
			\end{enumerate}
		\end{abstract}
		\thispagestyle{empty}
	\end{titlepage}
\maketitle
\newcommand{\ThmMainLi}
{
	
	For any constant $0<\epsilon<1/2$, any algorithm that finds a $(1/2+\epsilon)$-approximation for maximum independent set in the CONGEST model requires $\Omega(n/\log^3 n)$ rounds.
	
}

\newcommand{\ThmMainQuad}
{
	
	For any constant $0<\epsilon<1/4$, any algorithm that finds a $(3/4+\epsilon)$-approximation for maximum independent set in the CONGEST model requires $\Omega(n^2/\log^3 n)$ rounds.
	
}


\section{Introduction and Related Work}

Consider a network of $n$ nodes, where each has a unique $O(\log n)$-bit identifier, and they can communicate with each other via synchronized communication rounds. In each round, each node can send a (possibly different) $O(\log n)$-bit message to each of its neighbors. The task of the nodes is to compute some function of the network (e.g., its diameter, the value of a maximum independent set, etc.), while minimizing the number of communication rounds. This model is well known as the CONGEST model, and it is one of the major models of theoretical distributed graph algorithms~\cite{Peleg:book00}. 

In recent years, our understanding of the complexity of some problems in the CONGEST model has been substantially improving, thanks to a fruitful technique for proving lower bounds via reductions to two-party communication complexity. This technique, that was implicitly introduced by Peleg and Rubinovich~\cite{PelegR00,PelegR99}, and explicitly and formally defined by Das Sarma et al.~\cite{SarmaHKKNPPW12,SarmaHKKNPPW11}, was used to prove many lower bounds for fundamental graph problems, such as distance computations~\cite{AbboudCHK16,FrischknechtHW12,HolzerP14}, minimum spanning tree~\cite{PelegR99,Elkin06,SarmaHKKNPPW11}, minimum cut~\cite{GhaffariK13,SarmaHKKNPPW11}, minimum vertex cover~\cite{Censor-HillelKP17}, constructing and verifying spanners~\cite{Censor-HillelD18,Censor-HillelKP16}, subgraph detection~\cite{DruckerKO13,FischerGKO18,GonenO17}, approximate max-clique~\cite{CzumajK18}, hardness of distribued optimization~\cite{BachrachCDELP19}, distributed random walks~\cite{NanongkaiSP11}, and more (a complete list of papers is infeasible). 

While the two-party communication complexity model is already a very successful source for reductions, a natural question is whether using more players can bring us even closer to a satisfactory understanding of the CONGEST model. In~\cite{DruckerKO13}, the authors use a reduction to the three-party set-disjointness problem to show a lower bound for triangle detection in the CONGEST-Broadcast model, where unlike in the CONGEST model, the nodes are not allowed to send different messages to different neighbors in each round, and they can only broadcast one $O(\log n)$-bit message to all their neighbors in each round. This lower bound doesn't translate to the CONGEST model (In fact, no lower bound better than $1$ round is known for triangle detection in the CONGEST model~\cite{abs-1711-01623,FischerGKO18}). Whether multi-party communication complexity is of any use to show stronger results for the CONGEST model has remained open.

In this work, we introduce a framework of reductions to multi-party communication complexity. Our framework enables us to show improved hardness results for maximum independent set in the CONGEST model. Recently, Bachrach et al.\cite{BachrachCDELP19} showed that any algorithm that finds a $(5/6+\epsilon)$-approximation for maximum independent set must spend $\Omega(n/\log^6 n)$ rounds. Furthermore, they showed that finding a $(7/8+\epsilon)$-approximation requires  $\Omega(n^2/\log^7 n)$ rounds, which is nearly tight, as any problem can be solved in $O(n^2)$ rounds in the CONGEST model. Our results:

\begin{theorem}\label{thm:2MIS}
	\ThmMainLi
\end{theorem}
\begin{theorem}\label{thm:3/4MIS}
	\ThmMainQuad
\end{theorem}

While our results are not necessarily tight, we hope that our technique could pave the way for more and
stronger lower bounds in the CONGEST model. One important property of our technique is that it doesn’t suffer from the same limitations as the two-party framework, on which we elaborate later. We note our results hold even against randomized algorithms that succeed with probability
$p\geq 2/3$, and even for constant diameter graphs. The hard instances that are used to prove Theorems~\ref{thm:2MIS}
and~\ref{thm:3/4MIS} are weighted graphs, but we can extend our arguments for unweighted graphs as well by losing a logarithmic factor in the lower bounds (in terms of the number of rounds), as explained in Remark~\ref{remarkUnweighted}.

To prove Theorems~\ref{thm:2MIS}
and~\ref{thm:3/4MIS} we use reductions to $t$-party communication complexity where we use $t=O(1/\epsilon)$ players. For $t=2$, our constructions are similar to the ones presented in~\cite{BachrachCDELP19}, and can be viewed as simplified versions of them. While we also get a minor improvement in terms of the number of rounds (and not only in terms of approximations) it is worth pointing out that the improvements we get in terms of the number of rounds are not artifacts of the multi-party construction, but rather are artifacts of the simplifications compared to~\cite{BachrachCDELP19}.

~\\\noindent\textbf{The two-party reduction technique in a nutshell:} The vast majority of these reductions rely on the high communication complexity of the two-party set-disjointness problem. In the set-disjointness problem, there are two players, Alice and Bob, who receive two input strings $x,y\in\{0,1\}^k$, where Alice receives $x$, Bob receives $y$, and they wish to know if their strings intersect. That is, they wish to know if there is an index $i\in [k]$, such that $x_i=y_i=1$. It is well known that Alice and Bob must exchange $\Omega(k)$ bits in order to solve set-disjointness~\cite{KushilevitzN:book96,Razborov92}. A reduction to the set-disjointness problem is usually applied in the following manner. Assume that we want to prove a lower bound in the CONGEST model for deciding whether an input graph has some property $P$ (for example, $P$ can be the property of having a diameter less than $5$). Alice and Bob, before they start the protocol for set-disjointness, decide on a fixed graph construction $G=(V,E)$, and a partition of $G$ into two graphs $G_A=(V_A,E_A)$ and $G_B=(V_B,E_B)$, where $G_A$ is owned by Alice, and $G_B$ is owned by Bob. Then, each of the players adds some edges to their own graph based on their input string, where Alice adds edges based on $x$, Bob adds edges based on $y$, and the set of cut edges between $G_A$ and $G_B$ is fixed and doesn't depend on the input strings. The construction is defined in a way such that the strings $x$ and $y$ are disjoint if and only if the graph, with the additional edges based on the input strings, has the property $P$. Hence, in order to find whether $x$ and $y$ are disjoint, Alice and Bob can simulate a CONGEST algorithm that checks whether an input graph has the property $P$. Simulating the CONGEST algorithm in the two-party communication complexity model is done as follows. For messages that are sent on edges in $G_A$, Alice can simulate these messages without any communication with Bob. Similarly, Bob can simulate the messages that are sent on edges in $G_B$, without any communication with Alice. For the other messages, the ones that are sent on edges in the cut between $G_A$ and $G_B$, Alice and Bob exchange messages: if there is a message from a node in $G_A$ to a node in $G_B$, then Alice sends this message to Bob, and vice-versa. The conclusion that is made in such reductions is that if there is an $r$-round algorithm for deciding $P$ in the CONGEST model, and if the number of edges on the cut betwen $G_A$ and $G_B$ is at most $c$, then there is a protocol for solving two-party set-disjointness that uses $O(r\cdot c\cdot \log n)$ bits. This is because on each edge on the cut, the CONGEST algorithm sends $O(\log n)$ bits in each of the $r$ rounds. Hence, since the communication complexity of set-disjointness is $\Omega(k)$, we get a lower bound of $r=\Omega(\frac{k}{c\log n})$ rounds for any algorithm for deciding $P$ in the CONGEST model. That is, \emph{the smaller the cut, the stronger the lower bound}.

Censor-Hillel at al.\cite{Censor-HillelKP17} showed a small-cut two-party construction to prove a lower bound for maximum independent set, where they show that any algorithm for finding or computing the optimal value of a maximum independent set must spend $\Omega(n^2/\log^2 n)$ rounds. An independent set in a graph is a subset of the nodes where no two nodes in the subset are neighbors. A maximum independent set in a (possibly weighted) graph is an independent set of maximum total weight, where by total we mean the sum of the weights of the nodes in the independent set. Independent sets play vital role in theoretical and practical computer science, and the problem of computing exact or approximate maximum independent set has been attracting attention recently in the CONGEST model~\cite{BachrachCDELP19,Censor-HillelKP17,abs-1906-11524,BYCHGS17}. However, in terms of upper bounds, we are still unable to find fast algorithms that achieve approximation factors better than  $\Delta$, where $\Delta$ is the maximum degree of a node in the graph. If one is happy with a $\Delta$-approximation, or a $(1+\epsilon)\Delta$-approximation, then very fast and even sub-logarithmic algorithms exist~\cite{BYCHGS17,abs-1906-11524}. In terms of lower bounds, recently Bachrach et al.~\cite{BachrachCDELP19} built on the small-cut construction of~\cite{Censor-HillelKP17}, together with a very clever use of error-correcting codes, to prove near-linear hardness for $(5/6+\epsilon)$-approximation, and near-quadratic hardness for $(7/8+\epsilon)$-approximation.

~\\\noindent{\textbf{Limitations of the two-party framework:}} As pointed out by~\cite{BachrachCDELP19}, the two-party framework suffers from some limitations. Especially, but not only, when trying to use it to prove hardness of approximation. For example, for the maximum independent set problem, the two-party framework cannot show any lower bound against algorithms that achieve $(1/2)$-approximation. This is because Alice and Bob can compute the optimal solutions for maximum independent set in the graphs $G_A$ and $G_B$, without any communication, where Alice finds the optimal solution in $G_A$, and Bob finds the optimal solution in $G_B$. The maximum of the two values is always at least half of the optimal solution for $G$. Hence, by just exchanging the two values, which takes $O(\log n)$ bits of communication, Alice and Bob can find a $(1/2)$-approximation for maximum independent set. Since $O(\log n)$ bits can be sent in one round in the CONGEST model, no lower bound for this approximation factor can be shown
by  using the two-party framework. Similarly, the two party framework suffers from a limitation when trying to show a lower bound for $(3/2)$-approximation to minimum vertex cover, where the argument for vertex-cover is not trivial and was proved also  in~\cite{BachrachCDELP19}.

By using more players, the framework  doesn't suffer from the same limitations as in the two-party case. For example, with respect to approximating maximum independent set, the argument above translates only to a limitation of showing a  $(1/t)$-approximation. Hence, the more players we use, the less restrictive the limitations we get.

~\\\noindent{\textbf{The Challenge:}} Perhaps the first attempt that one would try in order to extend the two-party framework to the multi-party case is to use a reduction to the multi-party set-disjointness problem. In the multi-party set disjointness problem, there are $t$ players $p_1,\cdots, p_t$. Each receives a string $x^i\in \{0,1\}^k$, and they wish to know if the strings all intersect on the same index. That is, they wish to know if there is an index $m\in [k]$ satisfying $x^1_m=x^2_m=\cdots =x^t_m=1$. However, using a reduction to the multi-party set-disjointness problem is not a simple task, and as $t$ gets larger, the task becomes more challenging. This is because in the non-intersecting case, there are many sub-cases of pairwise intersections, and the reduction needs to take into account all these sub-cases. For example, if we try to extend the reduction of~\cite{BachrachCDELP19} to the multi-party set-disjointness problem, in the non-intersecting case, for every pair $i\neq j\in [t]$, whether the strings $x^i$ and $x^j$ are intersecting or not influences the size (or weight) of the maximum independent set. Hence, for the non-intersecting case, the reduction needs to take into account all the sub-cases of pairwise intersections, and, the more players we have, the more sub-cases we get, and the more infeasible the reduction becomes. 

In order to overcome this challenge, we use reductions to a certain \emph{promise pairwise disjointness} problem, rather than the  multi-party set-disjointness problem. In this promise pairwise disjointness problem, there are $t$ players each receiving a string $x^i\in \{0,1\}^k$, with the promise that the strings are either all intersecting in the same index, or pairwise disjoint. That is, in the non-intersecting case, for all pairs $i\neq j\in [t]$, it holds that $x^i$ and $x^j$ are disjoint. Most importantly, we don't have many sub-cases of pairwise intersections in the non-intersecting case. The communication complexity of this promise pairwise disjointness problem is $\Omega(k/t\log t)$~\cite{ChakrabartiKS03}, which is large enough for our needs, and we are able to use it to prove our results.

~\\\noindent{\textbf{Road-map:}} In Section~\ref{sec:prelims}, we begin with some useful definitions and tools. In Section~\ref{sec:framework}, we present our framework of reductions to the multi-party communication complexity model. The technical heart of the paper is provided in Sections~\ref{sec:linear} and~\ref{sec:-quadrlowerbound}, where we show our linear and quadratic lower bounds, respectively.

\section{Preliminaries}\label{sec:prelims}

\subsection{Multi-party Communication Complexity} 

Our lower bounds rely on reductions to the number-in-hand model of multi-party communication complexity. In the number-in-hand model, there are $t$ players, each is holding an input $x^i\in \{0,1\}^k$, and they wish to compute a joint function of their inputs $f(x^1,\cdots, x^t)$, where $t$ and $k$ are parameters of the model. The communication setting in the number-in-hand model can be defined in various ways. In this work we use the shared blackboard model (see also, for example,\cite{PhillipsVZ12}), where the players can exchange messages by writing them on a shared blackboard that is visible to all the players. The communication complexity in this model is formally defined as follows.

\begin{definition}\textbf{[Communication Complexity - Shared Blackboard]}\\
	Let $k\geq 1$, $t\geq 2$ be two integers, $f$ be a Boolean function  $f:\prod_{i=1}^t\{0,1\}^k\rightarrow \{\true,\false\}$, and $\mathcal{Q}$ be the family of protocols that compute $f$ correctly with probability at least $2/3$, in the shared blackboard model. Given $t$ inputs $x^1,\cdots, x^t$, denote by $\pi_{Q}(x^1,\cdots, x^t)$
	the transcript of a protocol $Q$ on the inputs $x^1,\cdots, x^t$, i.e. the sequence of
	bits that are written on the shared blackboard. The cost of a protocol $Q$ is
	\begin{align*}
	Cost(Q)=\max_{x^1,\cdots, x^t\in \{0,1\}^k }|\pi_{Q}(x^1,\cdots, x^t)|
	\end{align*}
	
	The communication complexity of $f$, denoted by $CC_f(k,t)$, is defined to be the minimum cost over all the possible protocols that compute $f$ correctly with probability at least $2/3$:
	\begin{align*}
	CC_f(k,t)=\min_{Q\in \mathcal{Q}}Cost(Q)
	\end{align*}
\end{definition}

Our lower bounds for the CONGEST model are achieved via reductions to the \emph{promise pairwise disjointness} function. For two strings $x,y\in \{0,1\}^k$, we say that $x$ and $y$ are disjoint if $\sum_{j=1}^k x_jy_j=0$. 

\begin{definition}\label{def:PairDisj}\textbf{[Promise Pairwise Disjointness]}\\
	Let $k\geq 1$, $t\geq 2$, and  $x^1,\cdots, x^t\in \{0,1\}^k$, with the promise that the strings $x^1,\cdots, x^t$ are either \emph{uniquely intersecting}, or pairwise disjoint. That is, either there is an $m\in [k]$ satisfying $x^1_m=x^2_m=\cdots=x^t_m=1$, or $x^i$ and $x^j$ are disjoint for all pairs $i\neq j\in [t]$. The promise pairwise disjointness function outputs $\true$ if the strings are pairwise disjoint, and $\false$ if they are uniquely intersecting\footnote{Throughout the paper, for any positive integer $k$, we denote by $[k]$ the set of positive integers $\{1,2,\cdots, k\}$.}.
\end{definition}

Chakrabarti et al.\cite{ChakrabartiKS03} proved that the communication complexity of the promise pairwise disjointness function in the shared blackboard model is $\Omega(k/t\log t)$.

\begin{theorem}\label{thm:PairDisj}\textbf[Theorem 2.5 in\cite{ChakrabartiKS03}]\\
	Let $f$ be the promise pairwise disjointness function. It holds that $CC_f(k,t)=\Omega(k/t\log t)$.
\end{theorem}

\subsection{Large Distance Codes}

Our proofs use the tool of \emph{error-correcting codes} that was used in~\cite{BachrachCDELP19}. Let us define the notion of a \emph{code-mapping}. Here, we use a similar definition to the one given by Arora and Barak~\cite{AroraBarak} (Chapter 19, Definition 19.5, page 380, in~\cite{AroraBarak}).

\begin{definition}\label{def:ECC}\textbf{[Code-mapping]}\\
	Let $\Sigma$ be a finite set of symbols, called the alphabet. Fix three integers $d\geq 1$, $L\geq 1$ and $M\geq L$. For two strings $x,y\in \Sigma^M$, the distance of $x$ and $y$, denoted by $d(x,y)$, is equal to $|\{i\in [M]\mid x_i\neq y_i\}|$.
	
	A code-mapping with parameters $(L,M,d,\Sigma)$ is a function $\Co:\Sigma^{L}\to \Sigma^M$, such that for every $x\neq y\in \Sigma^{L}$, $d(\Co(x),\Co(y))\geq d$.
\end{definition}

Our proofs use the following Theorem that shows the existence of large-distance codes (Lemma 19.11 in ~\cite{AroraBarak}).

\begin{theorem}\label{thm: ECC}
	Let $\Sigma$ be an alphabet of size $q=|\Sigma|$. There is a code-mapping with parameters $(L,M,d,\Sigma)$, where $L\leq M\leq q$ and $d=M-L$.
\end{theorem}

One way to construct a code-mapping that proves Theorem~\ref{thm: ECC} is by the so called \emph{Reed-Solomon} code, which is a well-known algebraic construction for error-correcting codes. In our proofs we don't need the details of the construction, but only its existence. 

\section{Multi-Party Communication Complexity Reductions}\label{sec:framework}

In this section we show how to prove lower bounds for the CONGEST model via reductions to the shared blackboard model of multi-party communication complexity. Our framework extends the framework of~\cite{Censor-HillelKP17} for the 2-party case. In~\cite{Censor-HillelKP17}, the authors define the notion of a \emph{family of lower bound graphs} for the 2-party case. In this work, we extend this notion for any arbitrary number $t\geq 2$ of players.

\begin{definition}\label{def: LBgraphs}\textbf{[Family of Lower Bound Graphs]}\\
	Given two integers $k\geq 1$, $t\geq 2$, a boolean function $f:\prod_{i=1}^t\set{0,1}^{k}\to\set{\true,\false}$, and a graph predicate $P$, a family of graphs $\set{G_{\bar{x}}=(V,E_{\bar{x}},w_{\bar{x}})\mid \bar{x}=(x^1,\cdots, x^t)\in\prod_{i=1}^t \{0,1\}^k}$ is said to be a family of \emph{lower bound graphs with respect to $f$ and $P$} if there is a partition of the set of nodes $V=\dot{\bigcup}_{i=1}^t V^i$ for which the following properties hold:\footnote{Throughout the paper, we use the notation $V=\dot\bigcup_{i=1}^t V^i$ to emphasize that $\{V^i\}_{i\in [t]}$ is a partition of $V$.}
	\begin{enumerate}
		\item \label{ItemInLBGraphs: va}
		Only the weight of the nodes in $V^i$ and the existence of edges in $V^i\times V^i$ may depend on $x^i$;
		\item \label{ItemInLBGraphs: pandf}$G_{\bar{x}}$ satisfies the predicate $P$ iff $f(\bar{x})=\true$.
	\end{enumerate}
\end{definition}

The intuition behind the definition of a family of lower bound graphs is as follows. Given a function $f$ whose input is split among $t$ players $p^1,\cdots, p^t$, where $p^i$ receives a string $x^i\in \{0,1\}^k$, and given a family of lower bound graphs $\set{G_{\bar{x}}=(V,E_{\bar{x}},w_{\bar{x}})\mid \bar{x}\in\prod_{i=1}^t \{0,1\}^k}$ with respect to $f$ and some graph predicate $P$. In order for the players to compute the value $f(x^1,\cdots, x^t)$, they can construct the graph $G_{\bar{x}}$, where $\bar{x}=(x^1,\cdots, x^t)$, and check whether $G_{\bar{x}}$ satisfies the predicate $P$. Due to the first condition of Definition~\ref{def: LBgraphs}, each player $p^i$ can construct the graph induced by the nodes in $V^i$ without any communication with the other players. Due to the second condition of Definition~\ref{def: LBgraphs}, $G_{\bar{x}}$ satisfies the predicate $P$ if and only if $f(x^1,\cdots, x^t)=\true$. Hence, the problem of  deciding whether $G_{\bar{x}}$ satisfies $P$ is reduced to computing the value $f(x^1,\cdots, x^t)$.

Next, we prove the following reduction theorem, which is based on a standard simulation argument. This theorem extends the reduction theorem of~\cite{Censor-HillelKP17}  for the 2-party case (Theorem 1 in~\cite{Censor-HillelKP17}). Given a family of lower bound graphs and a graph $G_{\bar{x}}$ in it, we denote by $cut(G_{\bar{x}})$ the set of \emph{cut edges} of $G_{\bar{x}}$. That is, $cut(G_{\bar{x}})=E_{\bar{x}}\setminus (\bigcup_{i=1}^t V^i\times V^i)$.

\begin{theorem}
	\label{thm: general lb framework}
	Fix $k\geq 1$, $t\geq 2$, $f:\prod_{i=1}^t \{0,1\}^k\to\set{\true,\false}$, and a graph predicate $P$. If there is a family $\set{G_{\bar{x}}=(V,E_{\bar{x}},w_{\bar{x}})}$ of lower bound graphs w.r.t.~$f$ and $P$, then any algorithm for deciding $P$ in the CONGEST model with success probability at least $2/3$ requires $\Omega \left(\frac{CC_f(k,t)}{\size{cut(G_{\bar{x}})}\log |V|}\right)$ rounds.
\end{theorem}

\begin{proof}
	Let $ALG$ be a distributed algorithm in the CONGEST model that decides $P$ in $T$ rounds. We define a protocol for $f$ in the shared blackboard model, as follows. Let $\bar{x}=(x^1,\cdots, x^t)\in \prod_{i=1}^t\{0,1\}^k$ be the vector of inputs of the players $p^1,\cdots, p^t$, where $p^i$ receives the string $x^i$, in the shared blackboard model. Each player $p^i$ constructs the part of $G_{\bar{x}}$ for the nodes in $V^i$. This can be done by the first condition of Definition~\ref{def: LBgraphs}, and the fact that the $V^i$'s are disjoint. 
	
	The players $p^1,\cdots, p^t$ simulate $ALG$, where each player $p^i$ simulates the nodes in $V^i$, as follows. All the messages that are sent on edges in $V^i\times V^i$ are simulated by player $p^i$, without any communication with the other players. All the other messages, the ones that are sent on edges in $cut(G_{\bar{x}})=E_{\bar{x}}\setminus \left(\bigcup_{i=1}^t V^i\times V^i\right)$, are written on the shared blackboard. That is, whenever there is a message from some node in $V^i$ to some node in $V^j$ for $i\neq j\in [t]$, player $p^i$ writes this message on the shared blackboard, which is visible to all the other players. In particular, it is visible to $p^j$ who is simulating the nodes in $V^j$.
	
	After simulating the $T$ rounds of $ALG$, the players know whether $G_{\bar{x}}$ satisfies the predicate $P$, and by the second condition of Definition~\ref{def: LBgraphs}, this reveals the information about $f(\bar{x})$. Observe that the total number of bits that are written on the blackboard are $O(T|cut(G_{\bar{x}})|\log |V|)$. This is because an algorithm in the CONGEST model sends at most $O(\log |V|)$ bits on each edge in each round, and the only messages that are written on the blackboard are the ones that are sent on the edges in $cut(G_{\bar{x}})$. Hence, the communication complexity of $f$ is at most $O(T|cut(G_{\bar{x}})|\log |V|)$ and therefore, $T=\Omega \left(\frac{CC_f(k,t)}{\size{cut(G_{\bar{x}})}\log |V|}\right)$.
\end{proof}

Our hardness results use families of lower bound graphs with respect to the promise pairwise disjointness function and a gap predicate $P$. We formalize such families in Definition~\ref{def:LBG}. First, we formally define the notion of $\gamma$-approximation for maximum independent set.

\begin{definition}\textbf{[$\gamma$-approximation for maximum independent set]}\\
	Let $G=(V,E,w)$ be a vertex-weighted graph with weight function $w$, and let $OPT$ be the value of an optimal solution for maximum independent set.\footnote{Throughout the paper, for a subset of nodes $U\subseteq V$, we denote by $w(U)=\sum_{v\in U} w(v)$.} An independent set $I$ in $G$ is $\gamma$-approximation for maximum independent set if $w(I)\geq OPT/\gamma$.
\end{definition}

\begin{definition}\label{def:LBG}\textbf{[$\gamma$-approximate MaxIS family of lower bound graphs]}\\
	Fix $0\leq\gamma\leq 1, \beta>0$. Let $P$ be a graph predicate that distinguishes between graphs of maximum independent set of weight at least $\beta$, and graphs of maximum independent set of weight at most $\gamma\cdot \beta$. A family of graphs is called a $\gamma$-approximate MaxIS if it is a family of lower bound graphs with respect to the promise pairwise disjointness function and the graph predicate $P$.
\end{definition}

The following corollary follows from Theorems~\ref{thm:PairDisj} and~\ref{thm: general lb framework}.

\begin{corollary}\label{cor: hardness}
	Let $k\geq 1$, $t\geq 2$ be two integers. If there is a $\gamma$-approximate MaxIS family of graphs $\set{G_{\bar{x}}=(V,E_{\bar{x}},w_{\bar{x}})\mid \bar{x}\in\prod_{i=1}^t \{0,1\}^k}$, then any algorithm for $\gamma$-approximation of maximum independent set in the CONGEST model with success probability at least $2/3$ requires $\Omega (k/(t\log t\cdot \size{cut(G_{\bar{x}})}\log |V|))$ rounds.
\end{corollary}

\section{Linear Lower Bound}\label{sec:linear}

In this section we prove the following theorem.
\begin{theorem-repeat}{thm:2MIS}
\ThmMainLi
\end{theorem-repeat}



In order to prove Theorem~\ref{thm:2MIS}, we construct a $(1/2+\epsilon)$-approximate $MaxIS$ family of lower bound graphs $\set{G_{\bar{x}}=(V,E_{\bar{x}},w_{\bar{x}})\mid \bar{x}\in\prod_{i=1}^t \{0,1\}^k}$.

\subsection{The family of lower bound graphs}\label{ssec:-family_linear_fixed}

We start by describing a fixed graph construction $G=(V,E,w)$, and then we describe how to get from $G$ and a vector of strings $\bar{x}\in \prod_{i=1}^t \{0,1\}^k$ the graph $G_{\bar{x}}=(V,E_{\bar{x}},x_{\bar{x}})$, which gives a family of graphs $\set{G_{\bar{x}}=(V,E_{\bar{x}},w_{\bar{x}})\mid \bar{x}\in\prod_{i=1}^t \{0,1\}^k}$.

Our fixed graph construction $G$ contains $t$ copies of a fixed \emph{base graph} $H$. We start by describing the base graph $H$. 


\paragraph{Some notations.} Let $k,\alpha,\ell$ be three positive integers that are to be chosen later such that $(\ell+\alpha)^\alpha=k$, and $\ell\gg\alpha$. Let $\mathcal{C}$ be a code-mapping given by Theorem~\ref{thm: ECC} with parameters $(\alpha,\ell+\alpha,\ell,\Sigma)$, where $\Sigma=\{1,\cdots , \ell+\alpha\}$. Observe that $k= |\Sigma|^{\alpha}$. Hence, we order the elements in $\Sigma^{\alpha}$ by an arbitrary ordering, and for $m\in [k]$, we denote by $\mathcal{C}(m)$ the code-mapping of the $m$'th element in $\Sigma^{\alpha}$. 

\paragraph{Description of $H=(V_H,E_H)$.} The set of nodes $V_H$ contains a clique of size $k$, denoted by $A=\{v_1,...,v_k\}$, and $\ell+\alpha$ cliques, $C_1,\cdots, C_{(\ell+\alpha)}$, each of size $\ell+\alpha$. For each $h\in [\ell+\alpha]$, the nodes in $C_h$ are denoted by $C_h=\{\sigma_{(h,1)},\cdots, \sigma_{(h,\ell+\alpha)}\}$. We call the cliques $C_1,\cdots , C_{\ell+\alpha}$ the \emph{code gadget}, and we denote this set of nodes by $$Code=\bigcup_{h=1}^{\ell+\alpha}C_h$$ The reason that these cliques are called the code-gadget is as follows. Given a code-word $w\in \Sigma^{\ell+\alpha}$, we can represent $w$ by $\ell+\alpha$ nodes $u_1\in C_1,u_2\in C_2,\cdots ,u_{\ell+\alpha}\in C_{\ell+\alpha}$, where $u_h\in C_h$ corresponds to the $h$'th position in $w$. That is, $u_h=\sigma_{(h,w_h)}$, where $w_h$ is the value in the $h$'th position in $w$. For any $m\in [k]$, we denote by $Code_m$ the set of nodes that corresponds to the code-word $\C(m)\in \Sigma^{\ell+\alpha}$, and we connect $v_m\in A$ to all the nodes in $Code\setminus Code_m$. 

This concludes the description of $H$ (see also Figure 1 for an example). More formally, the graph $H=(V_H,E_H)$ is defined as follows. Given a clique $C$, we denote by $E(C)$ the set of all the possible edges between nodes in $C$.
\begin{align*}
&V_H=A\cup Code\\
&E_H=E(A)\cup \{\{v_m,u\}\mid v_m\in A, u\in Code\setminus Code_m\}\bigcup_{h=1}^{\ell+\alpha}E(C_h)
\end{align*}
\begin{figure}                
	\centering
	\includegraphics[width=10cm,height=8cm]{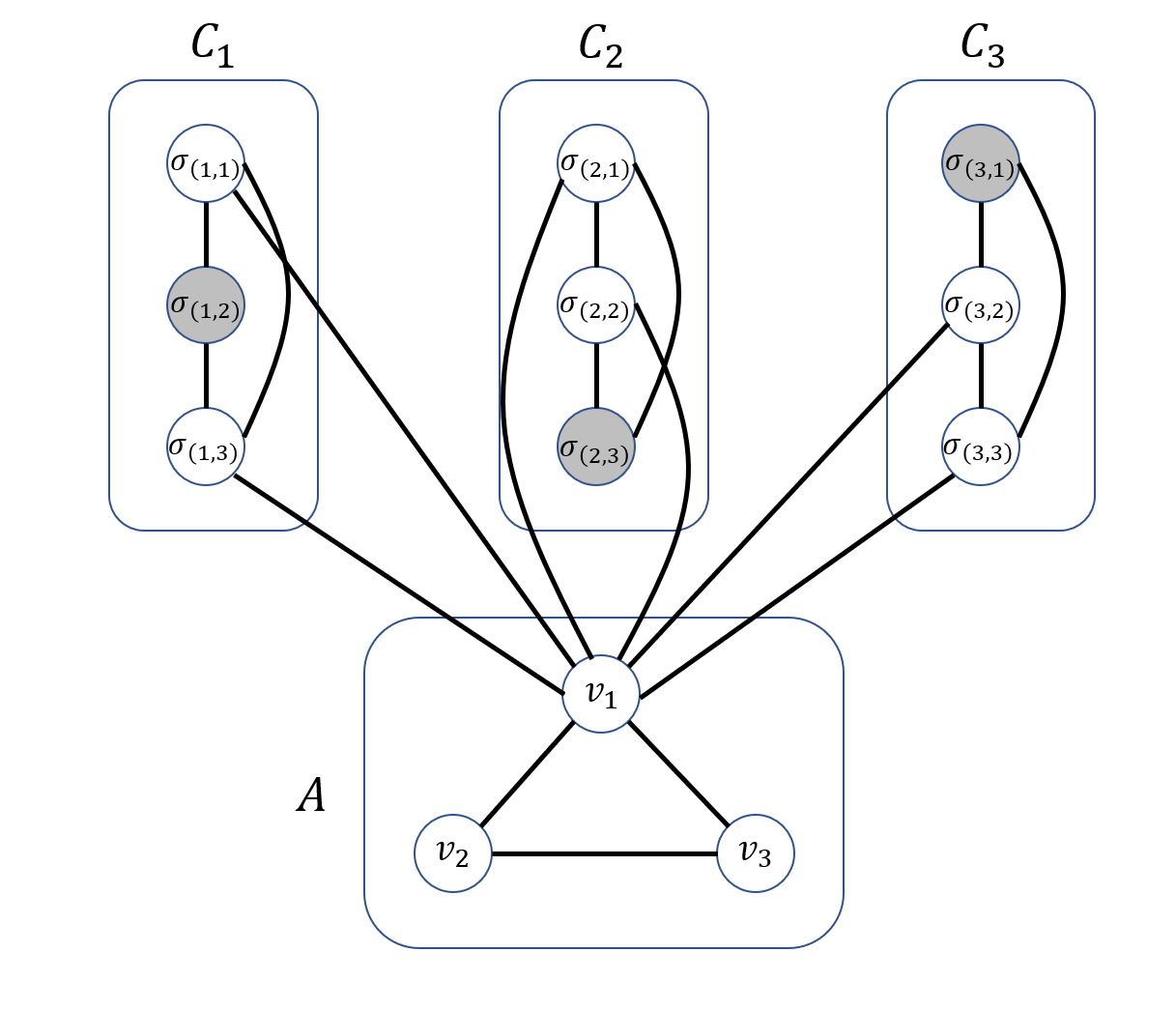}
	\label{fig:H}
	\caption{An example of the base graph $H$, where $\ell=2$, $\alpha=1$. $A$ is a clique of $k=(\ell+\alpha)^{\alpha}=3$ nodes, and there are $\ell+\alpha=3$ cliques $C_1,C_2$, and $C_3$, each of size $3$. In this example, we assume that the code-mapping of $1$,  $\C(1)=``2,3,1"$, and therefore, $v_1$ in connected to all the nodes in $Code=C_1\cup C_2\cup C_3$, except of the nodes in $Code_1= \{\sigma_{(1,2)},\sigma_{(2,3)},\sigma_{(3,1)}\}$. The other edges between $\{v_2,v_3\}\times Code$ are omitted in this figure, for clarity.}
\end{figure}
\paragraph{Obtaining the fixed graph construction $G$ from $H$:}

Now we are ready to describe the fixed graph construction $G=(V,E)$. Let $t\geq 2$. There are $t$ copies of $H$ in $G$, denoted by $H^1,\cdots, H^t$. In order to distinguish between nodes in different $H^i$'s, we add a superscript $i$ for the nodes in $H^i$. That is, for each $i\in [t]$, $H^i=(V^i,E^i)$ contains a clique and a code-gadget, where the clique is denoted by $A^i=\{v^i_1,\cdots, v^i_k\}$, the code-gadget is denoted by $Code^i$, the cliques in the code-gadget are denoted by $C^i_1,\cdots, C^i_{\ell+\alpha}$, and for any $h\in[\ell+\alpha]$, the nodes in $C^i_h$ are denoted by $C^i_h=\{\sigma^i_{(h,1)},\cdots, \sigma^i_{(h,\ell+\alpha)}\}$. Similarly, $Code^i_m$ denotes the set of nodes in $\bigcup_{i=1}^{\ell+\alpha} C^i_h$ that corresponds to the code-word $\mathcal{C}_m$. That is, let $w=\C(m)$, we have that $Code^i_m=\{\sigma^i_{(1,w_1)},\cdots, \sigma^i_{(\ell+\alpha,w_{\ell+\alpha})}\}$.

It remains to describe the connections between $H^i$ and $H^j$, for any $i\neq j\in [t]$. For any $h\in [\ell+\alpha]$, we add all the possible edges between  $C^i_h$ and $C^j_h$ \emph{except of the natural perfect matching between $C^i_h$ and $C^j_h$}, i.e., $\{\{\sigma^i_{(h,r)},\sigma^j_{(h,r)}\}\mid r\in[\ell+\alpha]\}$. More formally, we add the following edges for any $i\neq j\in [t]$ and any $h\in [\ell+\alpha]$,

\begin{align*}
\left\{\{u,v\}\mid u\in C^i_h, v\in C^j_h\right\}\setminus\left\{\{\sigma^i_{(h,r)},\sigma^j_{(h,r)}\}\mid r\in [\ell+\alpha]\right\}
\end{align*}
\begin{figure}                
	\centering
	\includegraphics[width=8cm,height=6cm]{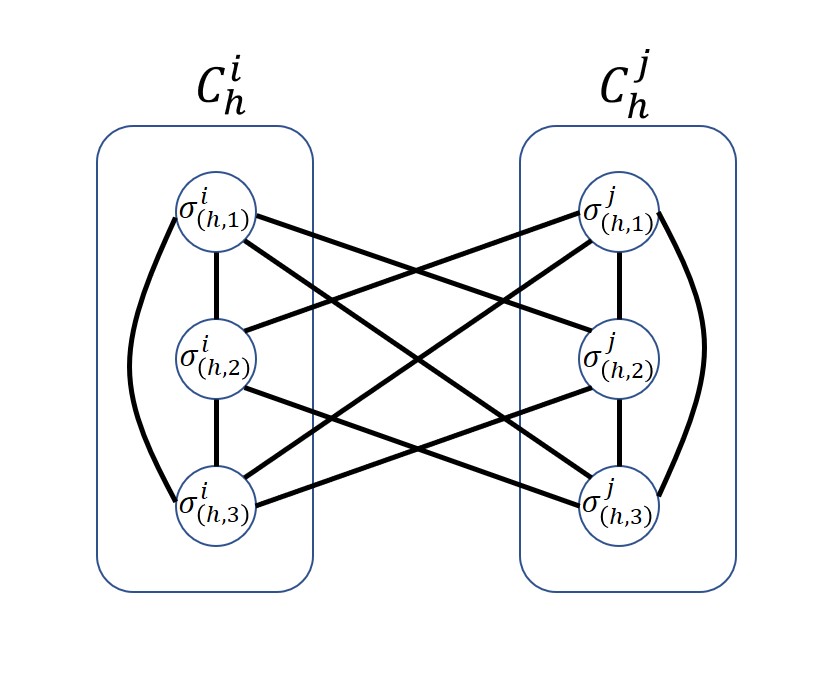}
	\label{fig:Con}
	\caption{An illustration for the connections between $C^i_h$ and $C^j_h$. In this example, $\ell+\alpha=3$. Observe that for any $r\in \{1,2,3\}$, $\sigma^i_{h,r}$ is connected to all the nodes in $C^j_h$ except of $\sigma^j_{h,r}$.}
\end{figure}
See also Figure 2 for an illustration of these connections. This concludes our fixed graph construction $G$, and we proceed to describing $G_{\bar{x}}$.
\begin{figure}
	\hspace*{-3.5cm}                 
	\centering
	\includegraphics[width=21cm,height=15cm]{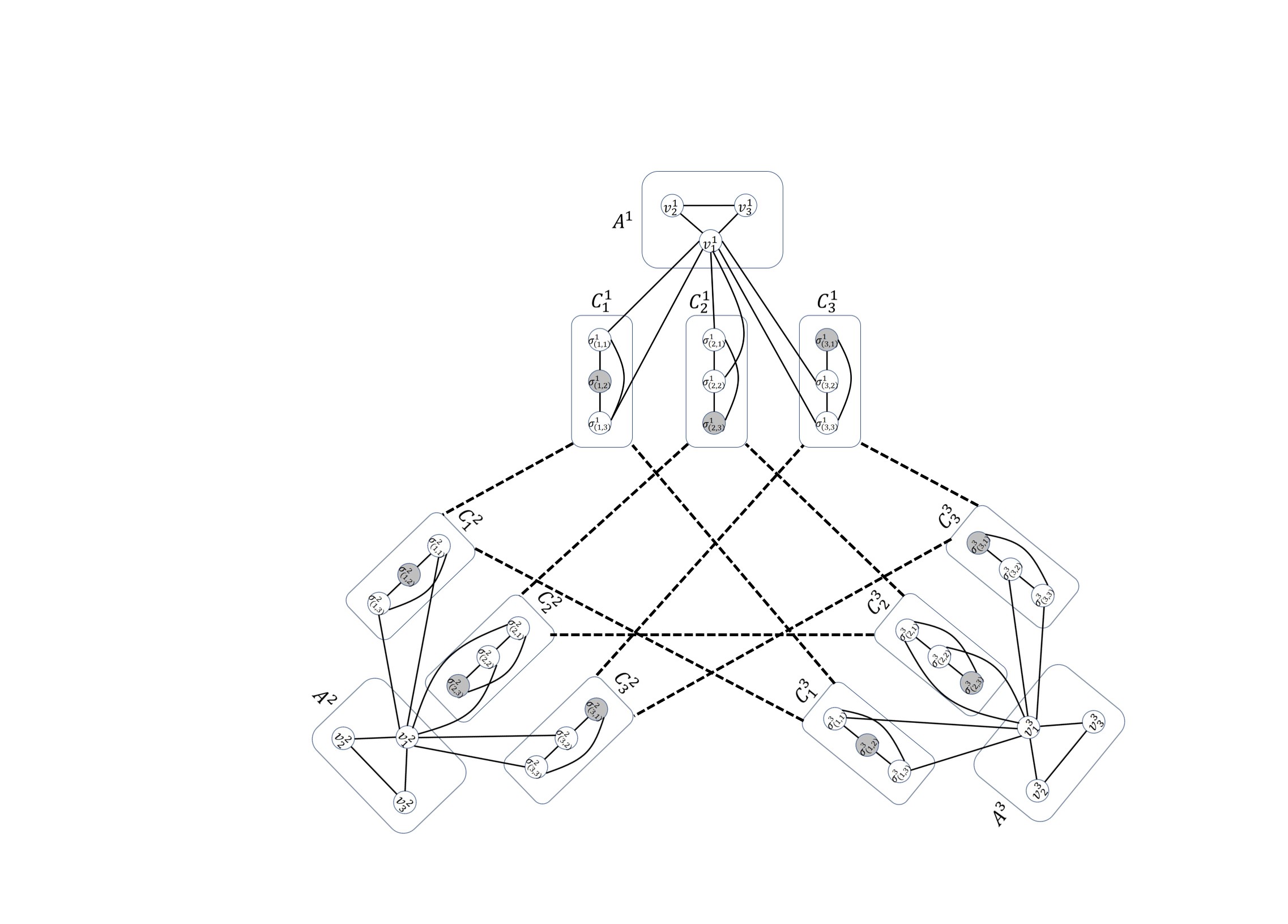}
	\label{fig:multi}
	\caption{Here, we have an illustration for a 3 players  construction, where $\ell=2,\alpha=1$, and $k=3$. Observe that for any $h,i\neq j\in [3]$, $C^i_h$ is connected by a dashed edge to $C^j_h$. This dashed edge represents all the connections between $C^i_h$ and $C^j_h$, as illustrated in Figure 2. Observe that $\{v^1_1,v^2_1,v^3_1\}\cup Code^1_1\cup Code^2_1\cup Code^3_1$ is an independent set.}
	\vspace*{-2cm}
\end{figure}
\paragraph{Obtaining $G_{\bar{x}}$ from $G$ and $\bar{x}$:}
Given $\overline{x}= (x^1,\cdots ,x^t)\in \prod_{i=1}^t \{0,1\}^k$. The graph $G_{\bar{x}}=(V,E,w_{\bar{x}})$ is defined as follows. The sets of nodes and edges of $G_{\bar{x}}$ are exactly as in $G$. The weights of nodes in $G_{\bar{x}}$ are defined as follows. Let $i\in [t]$, $m\in [k]$, and $v^i_m\in A^i$,

$$
w(v^i_m)=
\begin{cases}
\ell & \mbox{if } x^i_m=1\\
1 & \mbox{otherwise}\\
\end{cases}
$$

All the other nodes in $G_{\bar{x}}$ are of weight $1$. That is, for any $u\in V\setminus \bigcup_{i=1}^t A^i$, $w(u)=1$.

This concludes the description of $G_{\bar{x}}$. Before we proceed to proving that $G_{\bar{x}}$ is a family of lower bound graphs, we provide three useful properties of $G_{\bar{x}}$ that are used in the proof. 

\begin{property}\label{property: independence}
	For any $m\in [k]$, it holds that $\left(\bigcup_{i=1}^t  Code^i_m\right)\cup\{v^i_m\mid i\in [t]\}$ is an independent set.
\end{property}
\begin{proof}
	First, observe that the nodes in $\{v^i_m\mid i\in [t]\}$ are independent. This is because $v^i_m\in A^i$, and there are no edges between $A^i$ and $A^j$ for any $i\neq j$. There are also no edges between $A^i$ and $Code^j$, for any $i\neq j$. Furthermore, for any $i\in [t]$ and any $m\in [k]$, it holds that $\{v^i_m\}\cup  Code^i_m$ is an independent set. This is because $v^i_m$ is connected only to the nodes in $Code^i\setminus Code^i_m$. Finally, let $w=\C(m)$ be the code-mapping of $m$. Since for any $i\neq j$, we have that $Code^i_m=\{\sigma^i_{(1,w_1)},\cdots \sigma^{i}_{(\ell+\alpha,w_{\ell+\alpha})}\}$, and $Code^j_m =\{\sigma^j_{(1,w_1)},\cdots \sigma^{j}_{(\ell+\alpha,w_{\ell+\alpha})}\}$, and $\sigma^i_{(h,r)}$ is not connected to $\sigma^j_{(h,r)}$ for any $h,r\in [\ell+\alpha]$, we have that $\bigcup_{i=1}^t Code^i_m$ is an independent set. Hence, the union $\left(\bigcup_{i=1}^t  Code^i_m\right)\cup\{v^i_m\mid i\in [t]\}$ is an independent set. See also Figure 3 for an illustration.
\end{proof}

\begin{property}\label{property:largedistance}
	For any $i\neq j\in [t]$, and any $m_1\neq m_2\in [k]$, the bipartite graph $(Code^i_{m_1},Code^j_{m_2})$ contains a matching of size at least $\ell$.
	\begin{proof}
		Let $w^1=\C(m_1)$ be the code-mapping of $m_1$, and let the $w^2=\C(m_2)$ be the code-mapping of $m_2$. Given $h,r\in [\ell+\alpha]$, observe that $\sigma^i_{(h,r)}$ is connected to all the nodes in $C^j_h\setminus \{\sigma^j_{(h,r)}\}$. Hence, since the distance between $w^1$ and $w^2$ is at least $\ell$, there are at least $\ell$ positions $h\in [\ell+\alpha]$ for which $w^1_h\neq w^2_h$, and therefore, there are at least $\ell$ positions $h\in [\ell+\alpha]$ for which it holds that $\sigma^i_{(h,w^1_h)}$ is connected to $\sigma^j_{(h,w^2_h)}$, where $w^1_h$ is the $h$'th position in $w^1$ and $w^2_h$ is the $h$'th position in $w^2$.
	\end{proof}
\end{property}

\begin{property}\label{property:doubleIndices}
	Let $i\neq j\in [t]$, let $m_1\neq m_2\in [k]$, and let $I$ be any independent set. Let $w^1=\C(m_1)$ be the code mapping of $m_1$, and let $w^2=\C(m_2)$ be the code-mapping of $m_2$. The number of positions $h\in [\ell+\alpha]$ for which it holds that $\sigma^i_{(h,w^1_h)}\in I$ and $\sigma^j_{(h,w^2_h)}\in I$ is at most $\alpha$.
\end{property}
\begin{proof}
	By Property~\ref{property:largedistance}, the bipartite graph $(Code^i_{m_1},Code^j_{m_2})$ contains a matching of size at least $\ell$. Therefore, there are at least $\ell$ positions $h\in [\ell+\alpha]$ for which $I$ contains at most one of the nodes $\sigma^i_{(h,w^1_h)}$ and  $\sigma^j_{(h,w^2_h)}$. This leaves at most $\alpha$ other positions for which $I$ can contain both $\sigma^i_{(h,w^1_h)}$ and  $\sigma^j_{(h,w^2_h)}$.
\end{proof}

\subsection{$G_{\bar{x}}$ is a $(1/2+\epsilon)$-approximate $MaxIS$ family of lower bound graphs}

In this section we show that there is a constant  $t>2$ for which $G_{\bar{x}}$ is a $(1/2+\epsilon)$-approximate $MaxIS$ family of graphs. We start with a slightly weaker statement for $t=2$, which is later used in the proof for $t>2$.

\subsubsection{Warm-up: $t=2$}

In this section we prove the following lemma.

\begin{lemma}\label{lem:t=2}
	For $t=2$, and for any constant $\epsilon>0$, it holds that $\set{G_{\bar{x}}=(V,E_{\bar{x}},w_{\bar{x}})\mid \bar{x}\in\prod_{i=1}^t \{0,1\}^k}$ is a $(3/4+\epsilon)$-approximate $MaxIS$ family of lower bound graphs. 
\end{lemma}

For the rest of this subsection, we assume that $t=2$. Lemma~\ref{lem:t=2} is a corollary of Claims~\ref{claim:notDisj} and~\ref{claim:Disj}.

\begin{claim}\label{claim:notDisj}
	For any $g_{(x^1,x^2)}\in \set{G_{(x^1,x^2)}=(V,E_{(x^1,x^2)},w_{(x^1,x^2)})\mid (x^1,x^2)\in\prod_{i=1}^2 \{0,1\}^k}$, if $x^1$ and $x^2$ are not disjoint, then $g_{(x^1,x^2)}$ contains an independent set of weight at least $4\ell+2\alpha$.
\end{claim} 

\begin{proof}
	Since the sets are not disjoint, there is an $m\in [k]$ for which $x^1_m=x^2_m=1$. Therefore, the weight of each of the nodes $v^1_m$ and $v^2_m$ is $\ell$. By Property~\ref{property: independence}, the set $\{v^1_m\}\cup \{v^2_m\}\cup Code^1_m\cup Code^2_m$ is independent, and observe that its weight is $4\ell+2\alpha$.
\end{proof}

\begin{claim}\label{claim:Disj}
	For any $g_{(x^1,x^2)}\in \set{G_{(x^1,x^2)}=(V,E_{(x^1,x^2)},w_{(x^1,x^2)})\mid (x^1,x^2)\in\prod_{i=1}^2 \{0,1\}^k}$, if $x^1$ and $x^2$ are disjoint, then any independent set $I$ in $g_{(x^1,x^2)}$ is of weight at most $3\ell+2\alpha+1$.
\end{claim}

\begin{proof}
	The proof is by the following simple case analysis.
	\begin{enumerate}
		\item $I$ contains at most one node of weight $\ell$: In this case, the node of weight $\ell$ must be either in the clique $A^1$ or in the clique $A^2$. Assume without loss of generality that this node is in $A^1$. Observe that we can take at most one node of weight $1$ from $A^2$. Furthermore, since each of $Code^1$ and $Code^2$ is a union of $\ell+\alpha$ cliques, we cannot construct an independent set in $Code^1\cup Code^2$ of weight larger than $2(\ell+\alpha)$, it follows that the weight of $I$ cannot be larger than $3\ell+2\alpha+1$.
		\item $I$ contains two nodes of weight $\ell$: This implies that $I$ contains one node $v^1_{m_1}\in A^1$ of weight $\ell$ and another node $v^2_{m_2}\in A^2$ of weight $\ell$, where $m_1\neq m_2\in [k]$. Since the strings $x^1,x^2$ are disjoint, it must be the case that $m_1\neq m_2$. Furthermore, since $v^1_{m_1}$ is connected to the nodes in $Code^1\setminus Code^1_{m_1}$, and $v^2_{m_2}$ is connected to the nodes in $Code^2\setminus Code^2_{m_2}$, it remains to show that $|I\cap(Code^1_{m_1}\cup Code^2_{m_2})|\leq \ell+2\alpha$. By Property~\ref{property:largedistance}, $(Code^1_{m_1}, Code^2_{m_2})$ contains a matching of size at least $\ell$, and since $|Code^1_{m_1}\cup Code^2_{m_2}|=(2\ell+2\alpha)$, this implies that $|I\cap(Code^1_{m_1}\cup Code^2_{m_2})|\leq \ell+2\alpha$. To conclude, in this case, $I$ contains $2$ nodes of weight $\ell$ and $\ell+2\alpha$ nodes of weight $1$. In total, the weight of $I$ is $3\ell+2\alpha$.
	\end{enumerate}
	Notice that $I$ cannot contain more than $2$ elements of weight $\ell$ since the elements of weight $\ell$ form two disjoint cliques.
\end{proof}

\begin{proof}[\textbf{Proof of Lemma~\ref{lem:t=2}}]
	Claims~\ref{claim:notDisj} and ~\ref{claim:Disj} imply that $\set{G_{\bar{x}}=(V,E_{\bar{x}},w_{\bar{x}})\mid \bar{x}\in\prod_{i=1}^2 \{0,1\}^k}$ is a family of lower bound graphs with respect to the set disjointness function and the graph predicate that distinguishes between graphs of maximum independent set at least $4\ell+2\alpha$ and graphs of maximum independent set at most $3\ell+2\alpha+1$. 
	
	We set $\ell=\log k-\log k/\log\log k, \alpha=\log k/\log\log k$. Hence $(\ell+\alpha)^{\alpha}=k$ as desired. Since the dominating terms in the two cases are $4\ell$ and $3\ell$, it follows that for any constant $\epsilon>0$,  $\set{G_{\bar{x}}=(V,E_{\bar{x}},w_{\bar{x}})\mid \bar{x}\in\prod_{i=1}^2 \{0,1\}^k}$ is a $(3/4+\epsilon)$-approximate $MaxIS$ family of graphs\footnote{In fact, by slightly changing the parameters $\ell$ and $\alpha$, the claim holds for any $\epsilon=\Omega(1/\log k)$.}.
\end{proof}

\subsubsection{Hardness Amplification using $t>2$ Players}

In this section we prove the following lemma.

\begin{lemma}\label{lem:t>2}
	For any constant $\epsilon>0$, there is a constant $t>2$ for which it holds that $\set{G_{\bar{x}}=(V,E_{\bar{x}},w_{\bar{x}})\mid \bar{x}\in\prod_{i=1}^t \{0,1\}^k}$ is a $(1/2+\epsilon)$-approximate $MaxIS$ family of lower bound graphs. 
\end{lemma}

Lemma~\ref{lem:t>2} follows from Claims~\ref{claim:notDisjt>2} and~\ref{claim:Disjt>2}.

\begin{claim}\label{claim:notDisjt>2}
	For any positive integer $t$, and any $g_{\bar{x}}\in \set{G_{\bar{x}}=(V,E_{\bar{x}},w_{\bar{x}})\mid \bar{x}\in\prod_{i=1}^t \{0,1\}^k}$, if there is an $m\in [k]$ for which it holds that $x^1_m=\cdots= x^t_m=1$, then $g_{\bar{x}}$ contains an independent set of weight at least $t(2\ell+\alpha)$.
\end{claim}
\begin{proof}
	 Observe that for any $i\in [t]$, it holds that $w(v^i_m)=\ell$. Furthermore, by Property~\ref{property: independence}, $\left(\bigcup_{i=1}^t Code^i_m\right)\cup\{v^i_m\mid i\in[t]\}$ is an independent set, and it is of weight $2t\ell+t\alpha$.
\end{proof}


	
	


Before we proceed to the case in which the strings are pairwise disjoint, let us prove the following helper claim and a corollary of it.

\begin{claim}\label{claim:helper}
	For any positive integer $t$. Let $m_1,m_2,\cdots, m_t$ be any $t$ distinct values in $[k]$. For any independent set $I$, if $\{v^i_{m_i}\mid i\in [t]\}\subseteq I$, then $$\big|I\cap \big(\bigcup_{i=1}^t Code^i_{m_i}\big)\big|\leq \ell+\alpha t^2$$
\end{claim}
\begin{proof}
	Let us start with some notations. Let $w^i=\C(m_i)$ be the code-mapping of $m_i$. Hence, we have that  $Code^i_{m_i}=\{\sigma^i_{(1,w^i_1)},\cdots, \sigma^i_{(\ell+\alpha,w^i_{\ell+\alpha})}\}$. Furthermore, let $S=\{h\in [\ell+\alpha]\mid\sum_{i=1}^t|I\cap \sigma^i_{(h,w^i_h)}|\leq 1\}$, $\bar{S}=[\ell+\alpha]\setminus S$. That is, $S$ is the set of values $h\in [\ell+\alpha]$ for which the independent set $I$ contains at most one node in $\bigcup_{i=1}^t \{\sigma^i_{(h,w^i_h)}\}$. Finally, let $\psi^h_{i,j}$ be an indicator defined as follows.	
	\[
	\psi^h_{i,j} = \begin{cases}
	1 & \text{if } \sigma^i_{(h,w^i_h)}\in I \text{ and } \sigma^j_{(h,w^j_h)}\in I\\
	0 & \text{otherwise } 
	\end{cases}
	\]

	By Property~\ref{property:doubleIndices},  for any $i\neq j\in [t]$, it holds that $\sum_{h\in [\alpha+\ell]} \psi^h_{i,j}\leq \alpha$. Hence,
	\begin{align}
	&\big|I\cap \big(\bigcup_{i=1}^t Code^i_{m_i}\big)\big| = \sum_{i=1}^t \big|I\cap Code^i_{m_i}\big|=\sum_{i=1}^t \big|I\cap \big(\bigcup_{h=1}^{\ell+\alpha} \{\sigma^i_{(h,w^i_h)}\}\big)\big|\\
	&=\sum_{i=1}^t \sum_{h=1}^{\ell+\alpha}\big|I\cap \{\sigma^i_{(h,w^i_h)}\}\big|=\sum_{h=1}^{\ell+\alpha}\sum_{i=1}^t\big|I\cap \{\sigma^i_{(h,w^i_h)}\}\big|\\
	&=\left(\sum_{h\in S} \sum_{i=1}^t\big|I\cap \{\sigma^i_{(h,w^i_h)}\}\big|\right) + \left(\sum_{h\in \bar{S}} \sum_{i=1}^t\big|I\cap \{\sigma^i_{(h,w^i_h)}\}\big|\right)\\
	&\leq \left(\sum_{h\in S} \sum_{i=1}^t\big|I\cap \{\sigma^i_{(h,w^i_h)}\}\big|\right) + \left(\sum_{h\in \bar{S}} \sum_{i\neq j\in [t]} 2\psi^h_{i,j}\right)\\
	&=\left(\sum_{h\in S} \sum_{i=1}^t\big|I\cap \{\sigma^i_{(h,w^i_h)}\}\big|\right) + \left(\sum_{i\neq j\in [t]} \sum_{ h\in \bar{S} } 2\psi^h_{i,j}\right)\\
	&\leq \ell+\alpha+2\alpha\cdot t(t-1)/2\leq \ell+\alpha t^2
	\end{align}
	Where (1), (2), and (3) are straightforward. (4) holds because for any $h\in \bar{S}$, there are at least two indices $i\neq j\in [t]$, for which it holds that $\sigma^i_{(h,w^i_h)}\in I$ and $\sigma^j_{(h,w^j_h)}\in I$. (5) holds by changing the summation order of the second sum. (6) holds because for any $h\in S$, $\sum_{i=1}^t |I\cap \{\sigma^i_{(h,w^i_h)}\}|\leq 1$, and by Property 3, $\sum_{h\in [\ell+\alpha]} \psi^h_{i,j}\leq \alpha$.
\end{proof}

\begin{corollary}\label{cor:helper}
	For any positive integer $t$, let $m_1,m_2,\cdots, m_t$ be any $t$ distinct values in $[k]$. For any independent set $I$, if $\{v^i_{m_i}\mid i\in [t]\}\subseteq I$, then $$w(I)\leq (t+1)\ell+\alpha t^2$$
\end{corollary}
\begin{proof}
	Since each $v^i_{m_i}$ is connected to all the nodes in $Code^i\setminus Code^i_{m_i}$, we have that
	\begin{align*}
	&w(I)=w(I\cap (\bigcup_{i=1}^t A_i)) + w(I\cap (\bigcup_{i=1}^tCode^i))=\left(\sum_{i=1}^t w(v^i_{m_i})\right) + w(I\cap (\bigcup_{i=1}^tCode^i_{m_i}))\\&\leq t\ell+\ell+ \alpha t^2=(t+1)\ell +\alpha t^2
	\end{align*}
\end{proof}

\begin{claim}\label{claim:Disjt>2}
	For any positive integer $t$, and any $g_{\bar{x}}\in \set{G_{\bar{x}}=(V,E_{\bar{x}},w_{\bar{x}})\mid \bar{x}\in\prod_{i=1}^t \{0,1\}^k}$, if the strings $x^1,\cdots, x^t$ are pairwise disjoint, then the weight of any independent set is at most $(t+1)\ell+\alpha t^2$.
\end{claim}

\begin{proof}

	The proof is by induction on $t$, where the base case of $t=1$ is straightforward (even the case of $t=2$ was already proved in Claim~\ref{claim:Disj}). We assume correctness for $t-1$, and prove correctness for $t$.
	Let $I$ be an independent set in $g_{\bar{x}}$. Recall that $A^i$ is a clique and therefore $|I\cap A^i|\leq 1$. The proof is by the following case analysis. 
	\begin{enumerate}
		\item\label{case:firstLinear} There is some $i\in [t]$ for which it holds that $I\cap A^i$ is either empty, or contains a node of weight $1$: Observe that in this case, $w(I\cap V^i)\leq \ell+\alpha + 1$. This is because any independent set contains at most $\ell+\alpha$ nodes in $Code^i=V^i\setminus A^i$. Furthermore, by the inductive hypothesis on the graph induced by the nodes in $\bigcup_{j\in [t]\setminus \{i\}} V^j$, we have that
		$w(I)\leq t\ell+\alpha(t-1)^2+\ell+\alpha+1\leq (t+1)\ell+\alpha(t^2-2t+1)+\alpha+1< (t+1)\ell+\alpha(t^2)$.
		 Where the last inequality holds since $\alpha\geq 1$, and $t>2$.
		\item\label{case:secondLinear} For any $i\in [t]$, $I\cap A^i$ contains a node of weight $\ell$, denoted by $v^i_{m_i}$: This case is proved directly, without applying the inductive hypothesis, as follows. First, since the strings $x^1,\cdots, x^{t}$ are pairwise disjoint, it must be the case that for any $i\neq j\in [t]$, $m_i\neq m_j$. This is because $w(v^i_{m_i})=\ell$ if and only if $x^i_{m_i}=1$, and if $m_i=m_j$, it would imply that $x^i$ and $x^j$ are not disjoint. 
		Hence, by Corollary~\ref{cor:helper}, we have that $$w(I)\leq (t+1)\ell+\alpha t^2$$ As desired.
	\end{enumerate} 
\end{proof}

\begin{proof}[\textbf{Proof of Lemma~\ref{lem:t>2}}]
	
	Claims~\ref{claim:notDisjt>2} and ~\ref{claim:Disjt>2} imply that $\set{G_{\bar{x}}=(V,E_{\bar{x}},w_{\bar{x}})\mid \bar{x}\in\prod_{i=1}^t \{0,1\}^k}$ is a family of lower bound graphs with respect to the pairwise disjointness function and the graph predicate that distinguishes between graphs of maximum independent set at least $t(2\ell+\alpha)$ and graphs of maximum independent set at most $(t+1)\ell+\alpha\cdot t^2$. 
	
	Recall that $\ell=\log k-\log k/\log\log k, \alpha=\log k/\log\log k$. Which implies that the graph predicate distinguishes between independent sets of weight at least $2t(\log k - \log k/\log\log k+\log k/\log\log k)=2t\log k$ and independent sets of weight at most $(t+1)(\log k-\log k/\log\log k)+t^2(\log k/\log\log k)\leq (t+2)\log k$, for any constant $t$ and $k\gg t$. Hence, for any constant $\epsilon>0$, we choose $t=2/\epsilon$ (or the first integer larger than $2/\epsilon$, if it is not an integer). This implies that for any constant $\epsilon>0$, there is a constant $t$ for which $\set{G_{\bar{x}}=(V,E_{\bar{x}},w_{\bar{x}})\mid \bar{x}\in\prod_{i=1}^t \{0,1\}^k}$ is a $(1/2+\epsilon)$-approximate $MaxIS$ family of graphs. 
\end{proof}

\begin{proof}[\textbf{Proof of Theorem~\ref{thm:2MIS}}]
	
	Observe that $k=\Theta(n)$, where $n=|V|$. Furthermore, $\set{G_{\bar{x}}=(V,E_{\bar{x}},w_{\bar{x}})\mid \bar{x}\in\prod_{i=1}^t \{0,1\}^k}$ is a $(1/2+\epsilon)$-approximate $MaxIS$ family of graphs, where the partition of the set of nodes that is needed for Definition~\ref{def: LBgraphs} is $V=\bigcup_{i=1}^t V^i$. Hence, by Corollary~\ref{cor: hardness} and the fact that $\size{cut(G_{\bar{x}})}=t^2\log^2 k=\Theta(\log^2 k)$, any algorithm for finding a $(1/2+\epsilon)$-approximation for maximum independent set in the CONGEST model with success probability at least $2/3$ requires $\Omega (k/(t\log t\cdot \size{cut(G_{\bar{x}})}\log |V|))=\Omega (n/(t\log t\cdot \log^3n)=\Omega(n/\log^3 n)$ rounds.
\end{proof}

\begin{remark}\label{remarkUnweighted}
	\textup{While our hard instances in the proof of Theorem~\ref{thm:2MIS} are weighted, it is easy to extend the argument for unweighted graphs as well, by losing a logarithmic factor in the lower bound (in terms of the number of rounds), as follows. For every node $v$ of weight $\ell$, we replace $v$ by an independent set of size $\ell$, denoted by $I(v)$. For every node $u$ that is adjacent to $v$ in our construction, if $u$ is of weight $1$, we connect all the nodes in $I(v)$ to $u$. Otherwise, if $u$ is of weight $\ell$, it means that it is replaced by an independent set of size $\ell$, denoted by $I(u)$. We connect $I(v)$ to $I(u)$ by a bi-clique (a full bipartite graph). The proof that the converted construction yields a hardness of $(1/2+\epsilon)$-approximation follows from a similar case analysis to the one provided for the weighted case. Since the number of nodes in the unweighted construction in $n=\Theta(k\ell)=\Theta(k\log k)$ rather than $\Theta(k)$, in terms of the number of rounds, we lose a logarithmic factor in the lower bound compared to the weighted case.}
\end{remark}

\section{Quadratic Lower Bound}\label{sec:-quadrlowerbound}
In this section we prove the following theorem.
\begin{theorem-repeat}{thm:3/4MIS}
	\ThmMainQuad 
\end{theorem-repeat}

In order to prove Theorem~\ref{thm:3/4MIS}, we construct a $(3/4+\epsilon)$-approximate $MaxIS$ family of lower bound graphs $\set{F_{\bar{x}}=(V,E_{\bar{x}},w_{\bar{x}})\mid \bar{x}\in\prod_{i=1}^t \{0,1\}^{k^2}}$. Observe that unlike the previous section, the length of the strings in $\bar{x}$ is $k^2$ rather than $k$. In our graph construction, similarly to the previous section, $k=\Theta(n)$. Hence, having the length of the strings being $k^2$ allows us to achieve a near-quadratic lower bound. Our hard instances are weighted graphs, and we can extend our argument to unweighted graphs as well by losing a logarithmic factor in the lower bound (in terms of the number of rounds) in the same way as explained in Remark~\ref{remarkUnweighted}.
\subsection{The family of lower bound graphs}\label{ssec:-familyquad}
We begin with describing a fixed graph construction, $F=(V_F,E_F,w_F)$, and then we describe how to get from $F$ and a vector of strings $\bar{x}\in \prod_{i=1}^t \{0,1\}^{k^2}$ the graph $F_{\bar{x}}=(V,E_{\bar{x}},x_{\bar{x}})$. Let $G$ be the fixed graph construction defined in Section~\ref{ssec:-family_linear_fixed}. The fixed graph construction $F$ consists of exactly two copies of $G$, denoted by $G^1$ and $G^2$. Recall that $G=(V_G,E_G)$ where $V_G=\bigcup_{i=1}^t A^i\cup Code^i$. In order to distinguish between the sets of nodes that belong to $G^1$ and the sets of nodes that belong to $G^2$, we add an ordered pair as a superscript $(i,b)$, where $b\in \{1,2\}$ indicates whether the set is in $G^1$ or in $G^2$. That is, the set of nodes of $G^1$ is $V_{G^1}=\bigcup_{i=1}^t A^{(i,1)}\cup Code^{(i,1)}$, and the set of nodes of $G^2$ is $V_{G^2}=\bigcup_{i=1}^t A^{(i,2)}\cup Code^{(i,2)}$. Hence the set of nodes of $F$ is $V_F=\bigcup_{i=1}^t V^i$, where for any $i\in [t]$, we denote by 

\begin{align*}
	&V^i=V^{(i,2)}\cup V^{(i,2)}\\
	&V^{(i,1)} = A^{(i,1)}\cup Code^{(i,1)}\\
	&V^{(i,2)} = A^{(i,2)}\cup Code^{(i,2)}\\
	&A^{(i,1)}=\{v^{(i,1)}_{m}\mid m\in [k]\}\\
	&A^{(i,2)}=\{v^{(i,2)}_{m}\mid m\in [k]\}\\
	&Code^{(i,1)}=\bigcup_{h=1}^{\ell+\alpha} C^{(i,1)}_h\\
	&Code^{(i,2)}=\bigcup_{h=1}^{\ell+\alpha} C^{(i,2)}_h\\
	&C^{(i,1)}_h=\{\sigma^{(i,1)}_{(h,1)},\cdots, 	\sigma^{(i,1)}_{(h,\ell+\alpha)}\}\\
	&C^{(i,2)}_h=\{\sigma^{(i,2)}_{(h,1)},\cdots, \sigma^{(i,2)}_{(h,\ell+\alpha)}\}\\
	&Code^{(i,1)}_w=\{\sigma^{(i,1)}_{(h,w_h)}\mid h\in [\ell+\alpha]\}\\
	&Code^{(i,2)}_w=\{\sigma^{(i,2)}_{(h,w_h)}\mid h\in [\ell+\alpha]\}\\
\end{align*}

The weight function $w_F$ is defined as follows. For any $v\in V_F$,

$$
w_F(v)=
\begin{cases}
\ell & \mbox{if } v\in \bigcup_{i=1}^t A^{(i,1)}\cup A^{(i,2)}\\
1 & \mbox{otherwise}\\
\end{cases}
$$

That is, the weight of any node in the cliques $\bigcup_{i=1}^t A^{(i,1)}\cup A^{(i,2)}$ is $\ell$, and the weight of any node in the code-gadgets $\bigcup_{i=1}^t Code^{(i,1)}\cup Code^{(i,2)}$ is $1$. Observe that unlike the previous section, the weights of the nodes don't depend on the strings in  $\bar{x}$. See also Figures 4, 5, and 6, for illustrations. \begin{figure}                
	\centering
	\includegraphics[width=13cm,height=8cm]{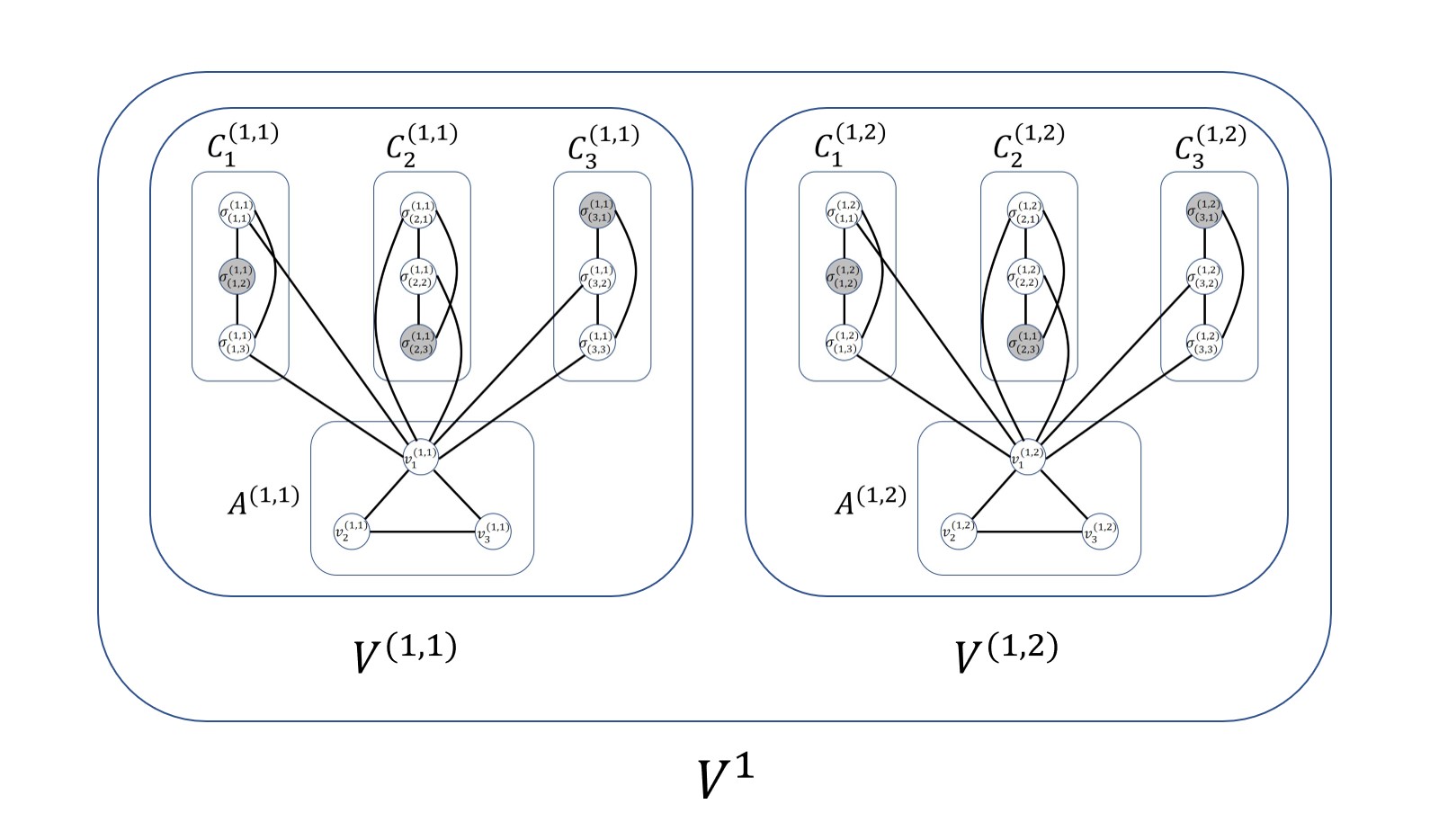}
	\label{fig:V1}
	\caption{An example of the graph induced by the nodes in $V^1$. As in the previous figures, $\ell=2,\alpha=1$ and $k=(\ell+\alpha)^{\alpha}=3$. $V^1$ contains two sets of nodes: $V^{(1,1)}$ which is in $G^1$, and $V^{(1,2)}$ which is in $G^2$. The graph induced by the nodes in $V^{(1,1)}$ has an identical topology to the graph induced by the nodes in $V^{(1,2)}$, and they are both identical to the topology of the base graph construction $H$ that was described is Section~\ref{ssec:-family_linear_fixed}. The reason that there is an ordered pair $(1,b)$, where $b\in \{1,2\}$ in a superscript in $V^{(1,1)}$ and $V^{(1,2)}$ is as follows. The first element in the pair indicates that these sets are parts of $V^1$, and the second element $b$ in the pair indicates that $V^{(1,b)}$ belongs to $G^{b}$. 
	Similarly, $V^{(1,1)}=A^{(1,1)}\cup Code^{(1,1)}=A^{(1,1)}\cup C^{(1,1)}_1\cup C^{(1,1)}_2\cup C^{(1,1)}_3$, and $V^{(1,2)}=A^{(1,2)}\cup Code^{(1,2)}=A^{(1,2)}\cup C^{(1,2)}_1\cup C^{(1,2)}_2\cup C^{(1,2)}_3$. As in the previous figures, the code-mapping of $1$, $\mathcal{C}(1)=``2,3,1"$, and therefore, $v^{(1,1)}_1$ is connected to all the nodes in $Code^{(1,1)}$ except of the nodes in $Code^{(1,1)}_1=\{\sigma^{(1,1)}_{(1,2)},\sigma^{(1,1)}_{(2,3)},\sigma^{(1,1)}_{(3,1)}\}$. Similarly, $v^{(1,2)}_1$ is connected to all the nodes in $Code^{(1,2)}$ except of the nodes in $Code^{(1,2)}_1=\{\sigma^{(1,2)}_{(1,2)},\sigma^{(1,2)}_{(2,3)},\sigma^{(1,2)}_{(3,1)}\}$. Some edges are omitted in this figure, for clarity. In particular, the existence of edges between $A^{(1,1)}$ and $A^{(1,2)}$ depends on the input string $x^1$, and these additional edges are illustrated in Figure 6. First, we illustrate in Figure 5 the full graph construction $F$ for $t=2$.}
\end{figure}

\begin{figure}                
	\centering
	\includegraphics[width=18cm,height=14cm]{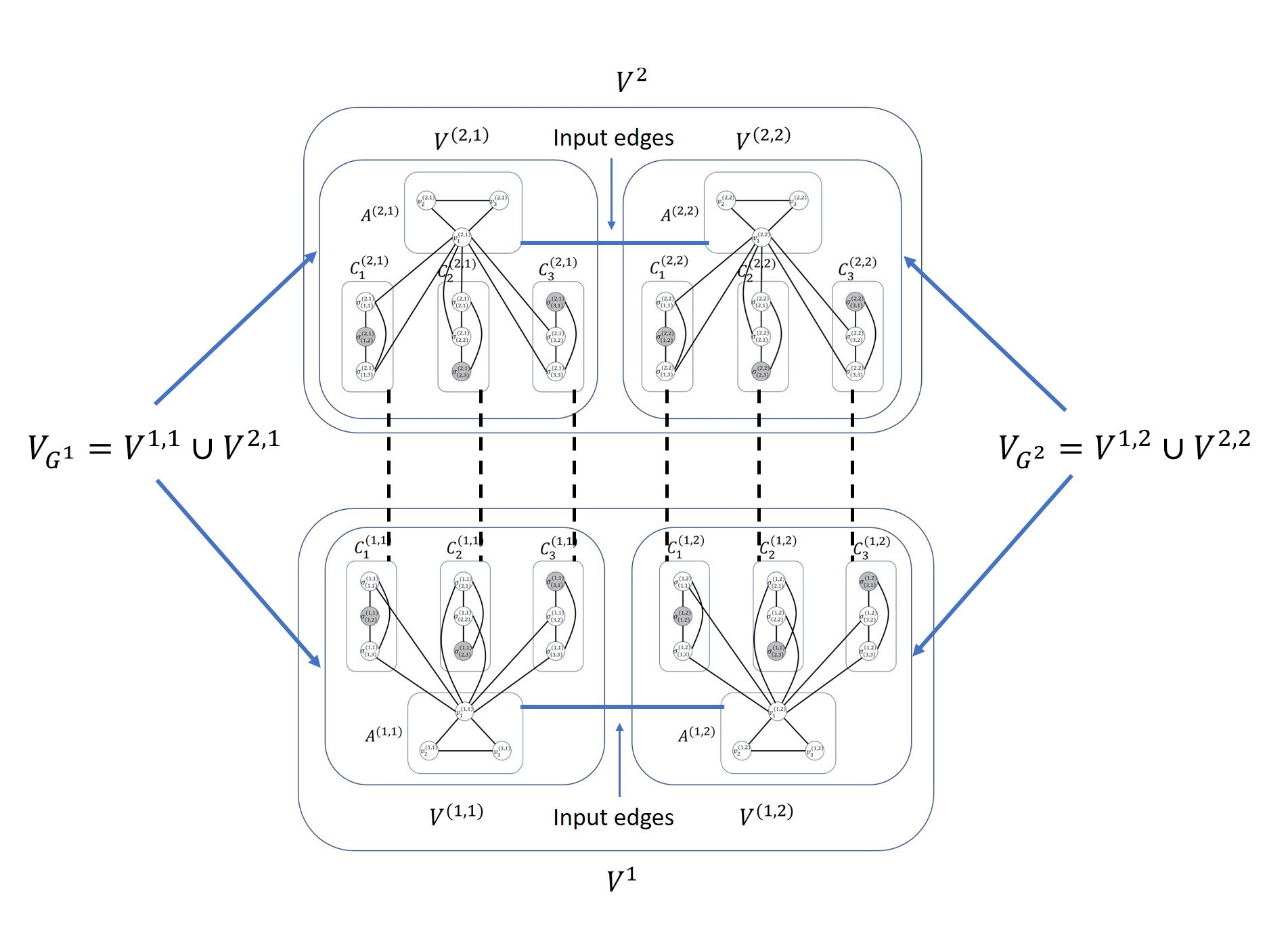}
	\label{fig:Quad2}
	\caption{An example for the full graph construction $F$ for $t=2$. Here, we also have $\ell=2$, $\alpha=1$ and $k=3$. There are two copies of the graph construction $G$ that was presented in Section~\ref{ssec:-family_linear_fixed}, $G^1$ and $G^2$, where the set of nodes of $G^1$ is $V_{G^1}=V^{(1,1)}\cup V^{(2,1)}$, and the set of nodes of $G^2$ is $V_{G^2}=V^{(1,2)}\cup V^{(2,2)}$. From the perspective of the two players, the set of nodes that is simulated by the first player is $V^1=V^{(1,1)}\cup V^{(1,2)}$, and the set of nodes that is simulated by the second player is $V^2=V^{(2,1)}\cup V^{(2,2)}$. For each $h\in [\ell+\alpha]=\{1,2,3\}$, and each $b\in \{1,2\}$, there is a dashed edge between $C^{(1,b)}_h$ and $C^{(2,b)}_h$, representing all the edges between $C^{(1,b)}_h$ and $C^{(2,b)}_h$, as explained in Figure~\ref{fig:Con}. All the edges in the graph $F$ are fixed and their existence doesn't depend on the input strings $\bar{x}=(x^1,x^2)$, \emph{except} of the edges between $A^{(1,1)}$ and $A^{(1,2)}$, that their existence depends on $x^1$, and the edges between $A^{(2,1)}$ and $A^{(2,2)}$, that there existence depends on $x^2$. The existence of these additional edges based on the input strings in illustrated in Figure 6. If $t=3$ instead of $2$, then the figure would have contained another set of nodes $V^3=V^{(3,1)}\cup V^{(3,2)}$, where $V^{(3,1)}$ is in $G^1$, and $V^{(3,2)}$ is in $G^2$, and the existence of edges between $A^{(3,1)}$ and $A^{(3,2)}$ depends on $x^3$.}
		
\end{figure}

\begin{figure}                
	\centering
	\includegraphics[width=13cm,height=8cm]{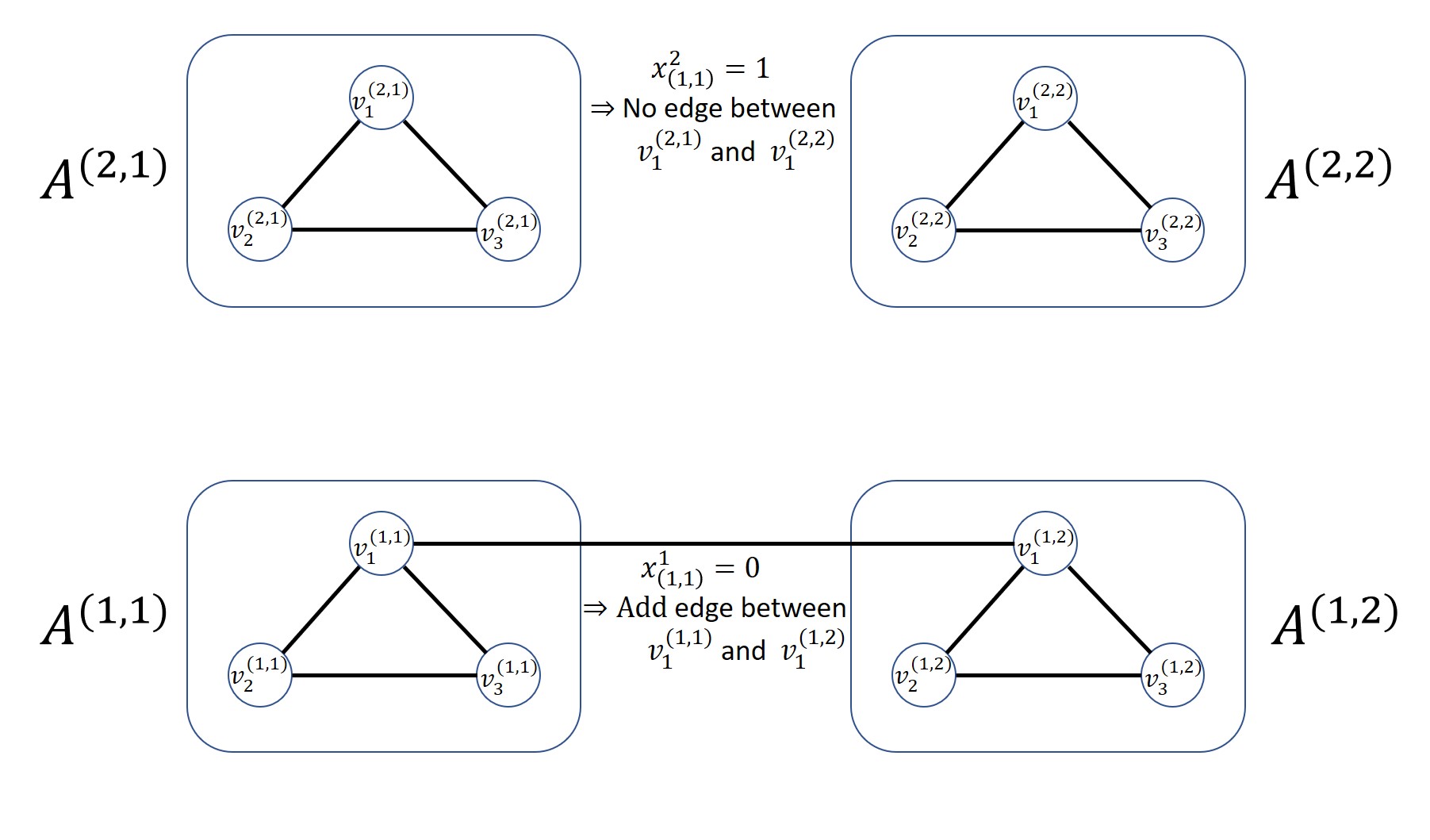}
	\label{fig:ConQuad}
	\caption{An illustration of the input edges. As in the other figures, we keep the example small and simple. Here, we also have $\ell=2$, $\alpha=1$, $k=3$ and $t=2$. We assume in this example that $\bar{x}=(x^1,x^2)$, where the  first bit in $x^1$ is $0$, and all the other bits are $1$. Furthermore, all the bits in $x^2$ are $1$. Hence, since the first bit in $x^1$ is indexed by $(1,1)$, and since its value is $0$, we add an edge between $v^{(1,1)}_1$ and $v^{(1,2)}_1$. Since all the bits in $x^2$ are $1$, we don't add any edges between $A^{(2,1)}$ and $A^{(2,2)}$.} 
\end{figure}

\paragraph{Obtaining $F_{\bar{x}}$ from $F=(V_F,E_F,w_F)$ and $\bar{x}$.} Let  $\bar{x}=(x^1,\cdots, x^t)\in \prod\limits _{i=1}^t \set{0,1}^{k^2}$. For any $x^i$, we index the $k^2$ positions in $x^i$ by $x^i_{(m_1,m_2)}$, for $m_1,m_2\in [k]$. The graph $F_{\bar{x}}$ is defined as follows. The set of nodes and the weight function remain exactly as in $F$. The set of edges contains all the edges in $F$, and the following edges in $A^{(i,1)}\times A^{(i,2)}$, for any $i\in [t]$. 
$$\{{v^{(i,1)}_{m_1}},v^{(i,2)}_{m_2}\mid x^i_{(m_1,m_2)}=0\}$$ 

That is, for any $i\in [t]$ and any $m_1,m_2\in [k]$, we add an edge between $v^{(i,1)}_{m_1}\in A^{(i,1)}$ and $v^{(i,2)}_{m_2}\in A^{(i,2)}$ if and only if $x^i_{(m_1,m_2)}=0$.
\subsection{$F_{\bar{x}}$ is a $(3/4+\epsilon)$-approximate $MaxIS$ family of lower bound graphs}

In this section we prove the  following lemma.
\begin{lemma}\label{lemma:-quad_family}
    For any constant $\epsilon>0$, there is a constant $t>2$ for which it holds that $\set{F_{\bar{x}}|\bar{x}\in \prod\limits _{i=1}^t \set{0,1}^{k^2}}$ is a $(3/4+\epsilon)$-approximate MaxIS family of lower bound graphs.
\end{lemma}

Lemma~\ref{lemma:-quad_family} is a corollary of Claims~\ref{claim:-quad_not_disjoint} and~\ref{claim:-quad_disjoint}.

\begin{claim}\label{claim:-quad_not_disjoint}
    For any $g_{\bar{x}}\in\set{F_{\bar{x}}|\bar{x}\in \prod\limits _{i=1}^t \set{0,1}^{k^2}}$, if there is a pair $(m_1,m_2)\in [k]\times [k]$ for which it holds that $x^1_{(m_1,m_2)}=x^2_{(m_1,m_2)}=\cdots =x^t_{(m_1,m_2)}=1$, then $g_{\bar{x}}$ contains an independent set of weight at least $4t\ell+2\alpha t$.
\end{claim}
\begin{proof}
     Consider the following set of nodes. 
     \[
     I=\bigcup\limits_{i=1}^t \{v^{(i,1)}_{m_1}\} \cup Code^{(i,1)}_{m_1}\cup \{v^{(i,2)}_{m_2}\}\cup Code^{(i,2)}_{m_2}
     \]
     First, by Property~\ref{property: independence}, it holds that both $\bigcup\limits_{i=1}^t \set{v^{(i,1)}_{m_1}} \cup Code^{(i,1)}_{m_1}$ and $\bigcup\limits_{i=1}^t \set{v^{(i,2)}_{m_2}} \cup Code^{(i,2)}_{m_2}$ are independent sets. Furthermore, the only possible edges between $\bigcup\limits_{i=1}^t \set{v^{(i,1)}_{m_1}} \cup Code^{(i,1)}_{m_1}$ and $\bigcup\limits_{i=1}^t \set{v^{(i,2)}_{m_2}} \cup Code^{(i,2)}_{m_2}$ are the ones in $\{\{v^{(i,1)}_{m_1},v^{(i,2)}_{m_2}\}\mid i\in [t] \}$. But since $x^1_{(m_1,m_2)}=x^2_{(m_1,m_2)}=\cdots =x^t_{(m_1,m_2)}=1$, none of the edges in  $\{\{v^{(i,1)}_{m_1},v^{i_2}_{m_2}\}\mid i\in [t] \}$ exists in the graph $F_{\bar{x}}$. The weight of $I$ is $|\bigcup_{i=1}^t w(\{v^{(i,1)}_{m_1},v^{(i,2)}_{m_2}\})|+|\bigcup_{i=1}^t w(Code^{(i,1)}_{m_1}\cup Code^{(i,2)}_{m_2})|= 2t\ell+2t(\ell+\alpha)=t(4\ell+\alpha)$, as desired.
\end{proof}

\begin{claim}\label{claim:-quad_disjoint}
For any $g_{\bar{x}}\in\set{G_{\bar{x}}|\bar{x}\in \prod\limits _{i=1}^t \set{0,1}^{k^2}}$, if the strings $x^1,x^2,\cdots ,x^t$ are pairwise disjoint, then the weight of any independent set in $g_{\bar{x}}$ is at most $3(t+1)\ell+3\alpha t^3$. 
\end{claim}
\begin{proof}
	The proof is by induction on $t$. For $t=1$, observe that $w(I\cap V^1)=w(I\cap (V^{(1,1)}\cup V^{(1,2)}))=w(I\cap (A^{(1,1)}\cup Code^{(1,1)}\cup A^{(1,2)}\cup Code^{(1,2)}))\leq 4\ell+2\alpha$.
	We assume correctness for $t-1$, and we prove correctness for $t$. Let $I$ be an independent set in $g_{\bar{x}}$. The proof is by the following case analysis. 
	\begin{enumerate}
		\item There is some $i\in [t]$, for which it holds that $|I\cap (A^{(i,1)}\cup A^{(i,2)})|\leq 1$: In this case, observe that $w(I\cap V^i)=w(I\cap (A^{(i,1)}\cup A^{(i,2)}\cup Code^{(i,1)}\cup Code^{(i,2)}))= w(I\cap (A^{(i,1)}\cup A^{(i,2)}))+w(I\cap (Code^{(i,1)}\cup Code^{(i,2)}))\leq \ell+2(\ell+\alpha)$. Hence, by applying the inductive hypothesis on the graph induced by the nodes in $\bigcup_{j\in [t]\setminus \{i\}}V^j$, we deduce that $w(I)= w(I\cap \bigcup_{j\in [t]\setminus \{i\}} V^j)+w(I\cap V^i)\leq 3t\ell+3\alpha (t-1)^3+3\ell+2\alpha<3(t+1)\ell+3\alpha t^3$.
		\item For all $i\in [t]$, it holds that $|I\cap (A^{(i,1)}\cup A^{(i,2)})|=2$: This case is proved without applying the inductive hypothesis, as follows. Fist, since $A^{(i,1)}$ and $A^{(i,2)}$ are cliques, there is one node in $I\cap A^{(i,1)}$ and one node in $I\cap A^{(i,2)}$. Denote these two nodes by $v^{(i,1)}_{m^1_i}\in A^{(i,1)}$ and $v^{(i,2)}_{m^2_i}\in A^{(i,2)}$, where $m^1_i,m^2_i\in [k]$. This implies that $v^{(i,1)}_{m^1_i}$ and $v^{(i,2)}_{m^2_i}$ are not connected by an edge. Since the strings $x^1,\cdots, x^t$ are pairwise disjoint, it must be the case that all the pairs in $\{(m^1_i,m^2_i)\mid i\in [t]\}$ are distinct. 
		
		We split the multiset of indices $\{m^1_i\mid i\in [t]\}$ into \emph{equivalence classes} by their value, where each class contains a set of indices of the same value. Observe that there are positive integers $r,q_1,q_2,\cdots,q_r$ satisfying $\sum_{j=1}^r q_j =t$, for which we can split $\{m^1_i\mid i\in [t]\}$ into $r$ equivalence classes $Q_1,\cdots , Q_r$, where $|Q_j|=q_j$. Let $s_i=\sum_{j=1}^i q_j$.\footnote{For example, if the multiset is $\{1, 1, 2, 3, 3, 3, 5\}$, then we have $r = 4, q_1 = 2, q_2 = 1, q_3 = 3, q_4 = 1$} Assume without loss of generality that 
		
		\begin{align*}
		&Q_1=\{m^1_1,\cdots, m^1_{s_1}\}\\
		&Q_2=\{m^1_{s_1+1},\cdots, m^1_{s_2}\}\\
		&Q_3=\{m^1_{s_2+1},\cdots, m^1_{s_3}\}\\
		&\vdots\\
		&Qr=\{m^1_{s_{r-1}+1},\cdots, m^1_{t}\}
		\end{align*}
		where
		\begin{align*}
		&m^1_1=\cdots= m^1_{s_1}\\ &m^1_{s_1+1}=\cdots= m^1_{s_2}\\
		&m^1_{s_2+1}=\cdots= m^1_{s_3}\\
		&\vdots\\
		&m^1_{s_{r-1}+1}=\cdots= m^1_{t}
		\end{align*}
		That is, we are assuming without loss of generality that $Q_1$ contains the first $s_1=q_1$ indices in $\{m^1_i\mid i\in [t]\}$,   $Q_2$ contains the next $q_2$ indices in $\{m^1_i\mid i\in [t]\}$, etc. This assumption is indeed without loss of generality because we can always split $\{m^1_i\mid i\in [t]\}$ into $r$ equivalence classes by their values, for some positive integer $r$, and our proof doesn't depend on the actual elements in each class. Since the pairs in $\{(m^1_i,m^2_i)\mid i\in [t]\}$ are distinct, it must be the case that 
		\begin{align*}
		&m^2_1\neq\cdots\neq m^2_{s_1}\\ &m^2_{s_1+1}\neq\cdots\neq m^2_{s_2}\\
		&m^2_{s_2+1}\neq\cdots\neq m^2_{s_3}\\
		&\vdots\\
		&m^2_{s_{r-1}+1}\neq\cdots\neq m^2_{t}
		\end{align*}
		
		The idea of the proof is to split the set of nodes into 3 disjoint sets, where the intersection of the independent set with each of the sets has small weight, as follows (we set $s_0=0$). 
		
		\begin{align*}
		&V=\bigcup_{i=1}^t V^{(i,1)}\cup V^{(i,2)}=\overbrace{\left(\bigcup_{j=0}^{r-1}  V^{(s_j+1,1)})\right)}^\text{First set}\cup \overbrace{\left(\bigcup_{j=1}^r (\bigcup_{i=s_{j-1}+2}^{s_j} V^{(i,1)})\right)}^\text{Second set}\cup\overbrace{\left(\bigcup_{j=1}^r (\bigcup_{i=s_{j-1}+1}^{s_j} V^{(i,2)})\right)}^\text{Third set}\\
		\end{align*} 
		In Propositions~\ref{prop:firstSet},~\ref{prop:secondSet}, and~\ref{prop:thirdSet}, we show that the intersection of the independent set with each of the three sets has small weight, and therefore, in total, the weight of the independent set is sufficiently small.
		\begin{proposition}\label{prop:firstSet}
			It holds that $$w\left(I\cap(\bigcup_{j=0}^{r-1}V^{(s_j+1,1)})\right)\leq (r+1)\ell+\alpha t^2$$
		\end{proposition}
		\begin{proof}
			Since $m^1_1,m^1_{s_1+1},\cdots, m^1_{s_{r-1}+1}$ are in different equivalence classes, they are distinct. Hence, by applying Corollary~\ref{cor:helper}, we have that
			
			\begin{align*}
			&w\left(I\cap(\bigcup_{j=0}^{r-1}V^{(s_j+1,1)})\right)\leq (r+1)\ell + \alpha r^2\leq (r+1)\ell + \alpha t^2
			\end{align*}
		\end{proof}
		\begin{proposition}\label{prop:secondSet}
			It holds that
			$$w\left(I \cap \bigcup_{j=1}^r (\bigcup_{i=s_{j-1}+2}^{s_j} V^{(i,1)})\right)\leq 2\ell(t-r)+\alpha (t-r)$$
		\end{proposition}
		\begin{proof}
			Since for any $i\in[t]$, $A^{(i,1)}$ is a clique, and $Code^{(i,1)}$ is a union of $\ell+\alpha$ cliques, we have that 
			
			\begin{align*}
			&w\left(\bigcup_{j=1}^r (\bigcup_{i=s_{j-1}+2}^{s_j} V^{(i,1)})\right)=\sum_{j=1}^r \sum_{i=s_{j-1}+2}^{s_j} w(I\cap V^{(i,1)})\\
			&=\sum_{j=1}^r \sum_{i=s_{j-1}+2}^{s_j} w\left(I\cap (A^{(i,1)}\cup Code^{(i,1)})\right)\\
			&\leq \sum_{j=1}^r \sum_{i=s_{j-1}+2}^{s_j} 2\ell+\alpha\\
			&=\sum_{j=1}^r 2\ell(|Q_j|-1)+\alpha(|Q_j|-1)\\
			&=\sum_{j=1}^r 2\ell(q_j-1)+\alpha(q_j-1)\\
			&=2\ell(\sum_{j=1}^r q_j)+2\ell(\sum_{j=1}^r-1)+\alpha(\sum_{j=1}^r q_j-1)\\
			&\leq 2\ell(t-r)+\alpha (t-r) 
			\end{align*}
			where the final inequality holds because $\sum_{j=1}^r q_j=t$.
		\end{proof}
	
		\begin{proposition}\label{prop:thirdSet}
			It holds that
			$$w\left(I \cap \bigcup_{j=1}^r (\bigcup_{i=s_{j-1}+1}^{s_j} V^{(i,2)})\right)\leq (t+r)\ell+\alpha t^3$$
		\end{proposition}
		\begin{proof}
			Since for any $j\in [r]$, it holds that $m^2_{s_{j-1}+1}\neq\cdots\neq m^2_{s_j}$. We can apply Corollary~\ref{cor:helper} on the graph induced by the nodes in $\bigcup_{i=s_{j-1}+1}^{s_j}V^{(i,2)}$ to deduce that
			\begin{align*}
			&w(I\cap (\bigcup_{i=s_{j-1}+1}^{s_j}V^{(i,2)}))\leq (|Q_j|+1)\ell+\alpha (|Q_j|)^2\\
			&=(q_j+1)\ell+\alpha {q_j}^2
			\end{align*}
			Hence, we have that
			\begin{align*}
			&w\left(I\cap\bigcup_{j=1}^r (\bigcup_{i=s_{j-1}+1}^{s_j} V^{(i,2)})\right)=\sum_{j=1}^r w(I\cap (\bigcup_{i=s_{j-1}+1}^{s_j}V^{(i,2)}))\leq\sum_{j=1}^r (q_j+1)\ell+\alpha q_j^2\\
			&\leq (t+r)\ell+\alpha t^2r\leq (t+r)\ell+\alpha t^3
			\end{align*}
		\end{proof}

		In total, 
		
		\begin{align*}
		&w(I)=w\left(I\cap(\bigcup_{j=0}^{r-1}V^{(s_j+1,1)})\right)+w\left(I\cap\bigcup_{j=1}^r (\bigcup_{i=s_{j-1}+2}^{s_j} V^{(i,1)})\right)+w\left(I\cap\bigcup_{j=1}^r (\bigcup_{i=s_{j-1}+1}^{s_j} V^{(i,2)})\right)\\
		&\leq (r+1)\ell + \alpha t^2+2\ell(t-r)+\alpha (t-r)+(t+r)\ell +\alpha t^3\\
		&\leq \ell(r+1+2t-2r+t+r)+3\alpha t^3\\
		&= \ell(3t+1)+3\alpha t^3\\
		\end{align*}
	As desired.
	\end{enumerate}
\end{proof}

\begin{proof}[\textbf{Proof of Lemma~\ref{lemma:-quad_family}}]
	
	Claims~\ref{claim:-quad_not_disjoint} and ~\ref{claim:-quad_disjoint} imply that $\set{G_{\bar{x}}=(V,E_{\bar{x}},w_{\bar{x}})\mid \bar{x}\in\prod_{i=1}^t \{0,1\}^{k^2}}$ is a family of lower bound graphs with respect to the pairwise disjointness function and the graph predicate that distinguishes between graphs of maximum independent set at least $4t\ell+2\alpha t$ and graphs of maximum independent set at most $3(t+1)\ell+3\alpha t^3$. 
	
	Recall that $\ell=\log k-\log k/\log\log k, \alpha=\log k/\log\log k$. Which implies that the graph predicate distinguishes between independent sets of weight at least $4t(\log k - \log k/\log\log k)+2\log k/\log\log k=4t\log k-2t\log k/\log\log k\geq 4(t-1)\log k$ and independent sets of weight at most $3(t+1)(\log k-\log k/\log\log k)+t^2(\log k/\log\log k)\leq 3(t+2)\log k$, for any constant $t$ and $k\gg t$. Hence, for any constant $\epsilon>0$, we choose $t=(3/4\epsilon)-1$ (or the first integer larger than $t=(3/4\epsilon)-1$, if it is not an integer). This implies that for any constant $0< \epsilon\leq 1/4$, there is a constant $t$ for which $\set{G_{\bar{x}}=(V,E_{\bar{x}},w_{\bar{x}})\mid \bar{x}\in\prod_{i=1}^t \{0,1\}^{k^2}}$ is a $(3/4+\epsilon)$-approximate $MaxIS$ family of graphs. 
\end{proof}

\begin{proof}[\textbf{Proof of Theorem~\ref{thm:3/4MIS}}]
	Observe that $k=\Theta(tn)=\Theta(n)$, where $n=|V|$. Furthermore, by Lemma~\ref{lemma:-quad_family},  $\set{G_{\bar{x}}=(V,E_{\bar{x}},w_{\bar{x}})\mid \bar{x}\in\prod_{i=1}^t \{0,1\}^{k^2}}$ is a $(3/4+\epsilon)$-approximate $MaxIS$ family of graphs, where the partition of the set of nodes that is needed for Definition~\ref{def: LBgraphs} is $V=\bigcup_{i=1}^t V^i$. Hence, by Corollary~\ref{cor: hardness}, the fact that the length of the strings is $k^2=\Theta(n^2)$, and the fact that $\size{cut(G_{\bar{x}})}=\Theta(t^2\log^2 k)=\Theta(\log^2 k)$, any algorithm for finding a $(3/4+\epsilon)$-approximation for maximum independent set in the CONGEST model with success probability at least $2/3$ requires $\Omega (k^2/(t\log t\cdot \size{cut(G_{\bar{x}})}\log |V|))=\Omega (n^2/(t\log t\cdot \log^3n)=\Omega(n^2/\log^3 n)$ rounds.
\end{proof}

\paragraph{Acknowledgments} We would like to thank Ami Paz for valuable discussions and useful comments. Yuval Efron was supported by the European Union's Horizon 2020 Research And Innovation Programme under grant agreement no.\ 755839, and also supported in part by the Binational Science Foundation (grant 2015803). Ofer Grossman was supported by a Fannie and John Hertz Foundation Fellowship, and an NSF GRFP fellowship. Finally, we thank the Simons institute for the generous hosting of the authors while performing this work.

\bibliographystyle{plain}
\bibliography{paper}

\end{document}